\documentclass[12pt, a4paper]{article}

\usepackage[margin=2.3cm]{geometry}
\usepackage{amsmath, amssymb, amsthm}
\usepackage{graphicx}
\usepackage{booktabs}
\usepackage{natbib}      % for author–year citations
\usepackage{setspace}
\usepackage{bbm}
\usepackage{algorithm} 
\usepackage{algpseudocode}
\usepackage{multirow}
\newcommand*{\E}{\mathbb{E}}
\newcommand\floor[1]{\lfloor #1 \rfloor}
\newcommand*{\Q}{\widehat{\mathbb Q}}
\DeclareMathOperator*{\argmax}{argmax}
\newtheorem{theorem}{Theorem}[section]
\newtheorem{corollary}{Corollary}[section]
\newtheorem{lemma}{Lemma}

\usepackage{hyperref}

\title{\bf Detecting Parameter Instabilities in Functional Concurrent Linear Regression}

\author{
    {\large Rupsa Basu$^{1}$ \quad and \quad Sven Otto$^{1,2}$}
}

\date{\today}

\begin{document}
\begin{singlespace}     % single-space the title page itself
\maketitle
\footnotetext[1]{University of Cologne, Institute of Econometrics and Statistics, Albertus-Magnus-Platz, 50923 Köln, Germany}
\footnotetext[2]{Corresponding author. Email: sven.otto@uni-koeln.de}

\end{singlespace}
\onehalfspacing     % restore double-spacing for the body
\begin{abstract}
We develop methodology to detect structural breaks in the slope function of a concurrent functional linear regression model for functional time series in $C[0,1]$. Our test is based on a CUSUM process of regressor-weighted OLS residual functions. To accommodate both global and local changes, we propose $L^2$- and sup-norm versions, with the sup-norm particularly sensitive to spike-like changes. Under Hölder regularity and weak dependence conditions, we establish a functional strong invariance principle, derive the asymptotic null distribution, and show that the resulting tests are consistent against a broad class of alternatives with breaks in the slope function. Simulation studies illustrate finite-sample size and power. We apply the method to sports data obtained via body-worn sensors from running athletes, focusing on hip and knee joint-angle trajectories recorded during a fatiguing run. As fatigue accumulates, runners adapt their movement patterns, and sufficiently pronounced adjustments are expected to appear as a change point in the regression relationship. In this manner, we illustrate how the proposed tests support interpretable inference for biomechanical functional time series.
\bigskip


\end{abstract}

\newpage

\section{Introduction}

\textbf{FTSA.} Functional time series analysis (FTSA) lies at the intersection of functional data analysis (FDA) and time series analysis, combining methods from functional spaces with time-dependent data. A functional time series is a sequence $\{X_i(t)\}_{i \in \mathbb Z }$, where for each time index $i \in \mathbb Z$, a real-valued function is defined on the domain $t \in [0,1]$, which is often interpreted as an internal time domain. In this manner, functional time series extend traditional time series $\{X_i\}_{i \in \mathbb Z }$ by incorporating a second index $t$ that captures internal structure, patterns, or cycles within each time period indexed by $i$.
This representation is particularly useful in applications where the continuum over $t$ for each observation $i$ is best regarded as a single element of a function space. See Figure~\ref{fig:snapshotData} for an illustration of sports data from running, where each cycle is indexed by $i$, and within-cycle variation occurs over $t \in [0,1]$. The workhorse space in FDA is the Hilbert space $L^2 [0,1]$ and classical tools such as covariance operators, spectral decompositions, and principal component analysis are naturally formulated in this Hilbert space setting. The statistical framework of functional data analysis (FDA) is a well-developed field in which various properties of such data as well as inference techniques for the same have been studied extensively. Comprehensive introductions to the FDA literature can be found in \cite{ramsay2002applied}, \cite{ferraty2006nonparametric}, \cite{horvath2012inference}, and \cite{wang2016functional}. 

\begin{figure}[htp]
    \centering
    \includegraphics[width=\linewidth]{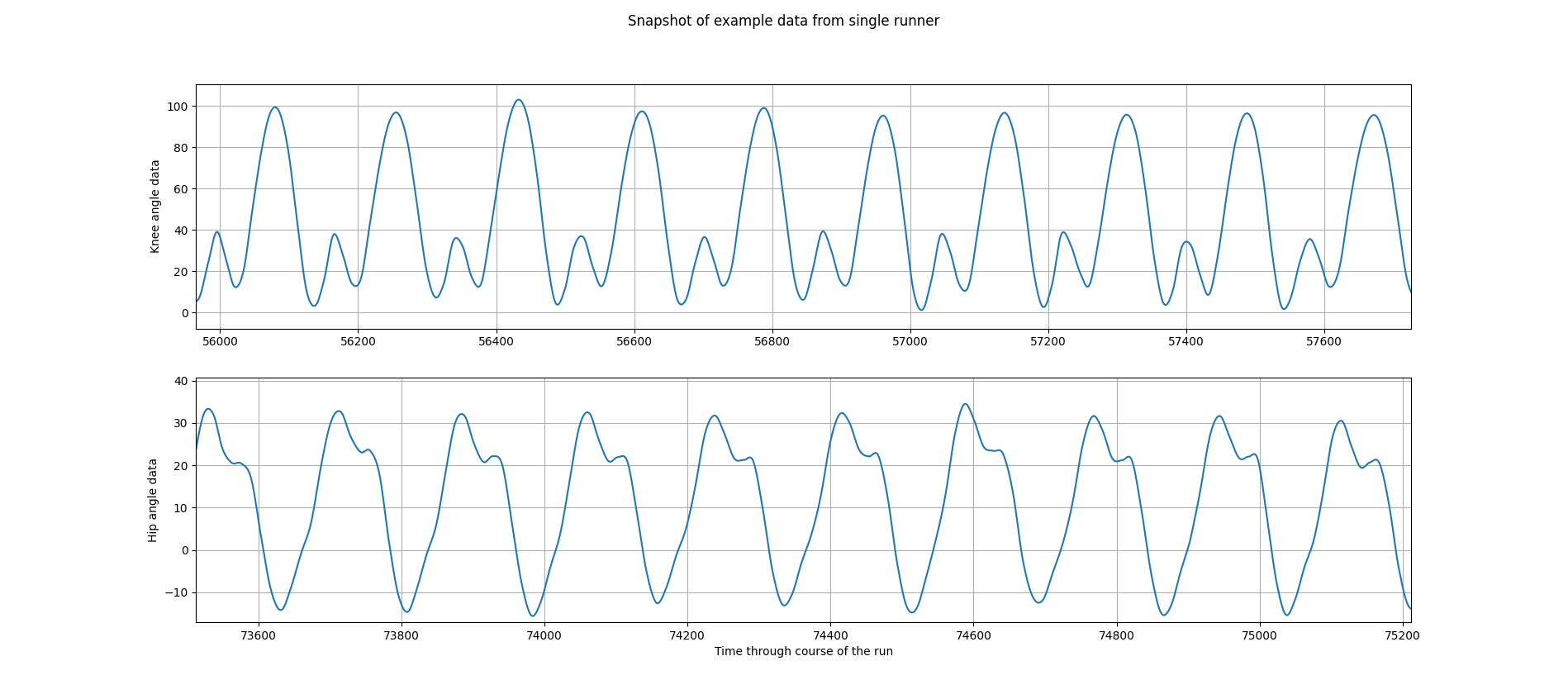}
    \caption{\textit{Top panel:} A snapshot of a few cycles from the knee angle data, denoted hereafter by $\{X_i (t)\}_{i \in \mathbb Z }$. \textit{Bottom panel:} Hip angle data, denoted by $\{Y_i (t)\}_{i \in \mathbb Z }$, for a single runner, where $i \in \mathbb Z $ denotes cycles and $t$ captures the internal cyclic structure within a single gait cycle.}
    \label{fig:snapshotData}
\end{figure}

\medskip

\noindent
\textbf{Going beyond $L^2$.} 
As an alternative to the popular $L^2[0,1]$ framework, we work in the separable Banach space of continuous functions $C[0,1]$, which permits access to well-defined pointwise values. We do this because $L^2$-distances can sometimes be small even for curves with markedly different shapes. Sharp spikes and impulses may not be captured well in $L^2$ as a result of integrating over the continuum $t$, whereas the sup-norm distinguishes them well. This has motivated recent developments in inference problems for FDA, including two-sample tests for equality of mean functions in the sup-norm (\citealt{dette2020functional}, \citealt{bastian2024multiplechangepointdetection}) and covariance functions (\citealt{dette2022detecting}), often in the context of change point detection.
This setting poses theoretical challenges because much of the existing FDA methodology is developed in the Hilbert space $L^2[0,1]$ and relies on $L^2$-geometry, which does not directly extend to $C[0,1]$ equipped with the sup-norm.
We refer the interested reader to Section~2 of \cite{dette2020functional} and the references therein for background on $C[0,1]$-valued random variables.
Further, $C[0,1]$-valued random variables can be viewed as jointly measurable stochastic processes on $[0,1]$ (see \citealt{Hsing_2015_Book}, Section~7). Accordingly, for such a random element $X$ we can work simultaneously in its dual role: pointwise, where $X(t)$ is a well-defined random variable for each $t\in[0,1]$, and functionally, where $X$ is an $L^2[0,1]$-valued random element.
Importantly, the $C[0,1]$ and $L^2 [0,1]$ frameworks are compatible because $C[0,1]$ embeds continuously into $L^2 [0,1]$ so boundedness in the sup-norm implies boundedness in $L^2$-norm and the standard Hilbert-space machinery remains available. In particular, if a $C[0,1]$-valued random variable has uniformly bounded finite moments, then its covariance function is finite and continuous on $[0,1]^2$. Consequently, by Mercer's theorem, the associated $L^2$-integral operator admits an orthonormal basis of continuous eigenfunctions and the Karhunen--Lo\`eve expansion holds in its natural $L^2$ sense.
In this paper, we make use of this unifying perspective and develop test statistics based on both the $L^2$- and the sup-norm: while the $L^2$-norm emphasizes average behavior of the functional curves, the sup-norm is sensitive to uniform deviations and localized departures. These differences are reflected in our simulations and in the wearable-sensor application, where fatigue-related adaptations may appear as sharp impulses before average deviations become pronounced.

\medskip

\noindent
\textbf{Our setup.} 
Our contribution focuses on change point inference problems for concurrent functional regression models. More precisely, we study a bivariate functional time series $\{X_i (t), Y_i (t)\}_{i \in \mathbb Z }$, where $X_i, Y_i \in C[0,1]$ are modeled using the concurrent linear regression framework,
\begin{align}
    Y_i (t) = \alpha_i (t) + \gamma_i (t) X_i (t) + \epsilon_i (t), \quad \quad i = 1, \dots, n, \ t \in [0,1], \label{flr_intro}
\end{align}
where $\epsilon_i \in C[0,1]$ denotes a mean zero functional regression error satisfying $\E (\epsilon_i (t) | X_i) = 0$ for all $t \in [0,1]$, and $\alpha_i, \gamma_i \in C[0,1]$ are intercept and slope parameter functions.
Such a concurrent regression model was first introduced by \cite{hastie1993varying} under the title \textit{varying coefficient model}, while the FDA literature later adopted the term concurrent functional linear model (\citealt{RS_2005_book}, \citealt{Zhang_et_al_2011}, \citealt{Petrovich_et_al_2023}). Related setups may be found in \cite{wu2000kernel} and \cite{csenturk2010functional}.

Change point detection in linear regression relationships for scalar time series data is well-known in the econometrics literature with contributions dating back to \cite{brown1975techniques}, \cite{kramer1988testing}, \cite{ploberger1992cusum}, \cite{andrews1993}, and \cite{bai1998}. In this work, we are interested in a functional extension of these classical approaches. The setting described in \eqref{flr_intro} is particularly relevant when the two-component functional time series is recorded simultaneously and potentially coupled. Such a concurrent linear regression model therefore provides a simple framework that takes into account the interdependent nature of the individual functional time series and captures the functional relationship between the dependent variable $Y_i(t)$ and the regressor function $X_i(t)$.

Our focus in this work is on testing the structural stability of the concurrent functional regression relationship over the time index $i = 1, \ldots, n$ against structural changes in the $C[0,1]$-valued slope parameter $\gamma_i$. In regression applications the slope function is typically the parameter of primary interest and therefore our tests are designed to have power against slope instability.
Under the standard exogeneity condition $\E (\epsilon_i (t) | X_i) = 0$ and mild moment conditions, $\alpha_i$ and $\gamma_i$ are identified via 
the implied population normal equations
\begin{align}
    \text{Cov} (X_i (t), Y_i (t)) &= \gamma_i (t) \text{Var} (X_i (t)), \label{eq:normal} \qquad
    \E(Y_i(t)) = \alpha_i (t) + \gamma_i (t)\, \E(X_i (t)).
\end{align}
The second relation in \eqref{eq:normal} linking $\alpha_i$ and $\gamma_i$ indicates that, if $\E(X_i (t)) \neq 0$ for some $t$, then changes in $\gamma_i$ across $i$ are generally accompanied by changes in $\alpha_i$.
Conversely, when the slope $\gamma_i$ is stable across $i$, variation in $\alpha_i$ corresponds to changes in the unconditional mean functions of $Y_i$ and $X_i$, which is a topic that has been extensively studied for functional time series (see \cite{aue2018detecting} for the $L^2$-norm and \cite{dette2020functional} for the sup-norm). We therefore focus on testing slope stability and do not address mean-change testing here.

\medskip

\noindent
\textbf{Practical motivation.} 
From a practical point of view, the methodology presented in this work is demonstrated on sports data collected from running athletes via body-worn sensors. As we will see, the data obtained from such sensors is recorded (fairly) continuously over time. In particular, gait movements like running involve repetitive motion, each cycle of which may be modeled as a realization of an underlying function; see Figure \ref{fig:snapshotData}. Moreover, the motion of lower-extremity joints is inherently interrelated, and we account for this by modeling the hip angle data as the dependent variable $Y_i(t)$ and the knee angle data as the regressor variable $X_i(t)$; see Figure \ref{fig:SlopeInterceptNull} for an illustration of $\gamma_i (t)$ and $\alpha_i (t)$ for this setting in a specific runner. In biomechanics or sports science, where movement patterns vary by individual, the representation of $\gamma_i (t)$ and $\alpha_i (t)$ allows personalized insights on the simultaneous movement of two joints; see, for example, the lower plots in Figure \ref{fig:SlopeInterceptNull}. For example, sharp dips in either the slope or intercept functions during parts of a cycle could suggest reduced coupling between hip and knee motion in certain phases of the gait cycle of a particular runner. Later in Section \ref{varyingcoeff-jointcoupling}, we demonstrate two examples, an experienced runner and a novice runner, and provide a discussion on the slope function relationship between their (right-) hip and knee angles.

\begin{figure}[t]
        \centering
        \includegraphics[width=\linewidth]{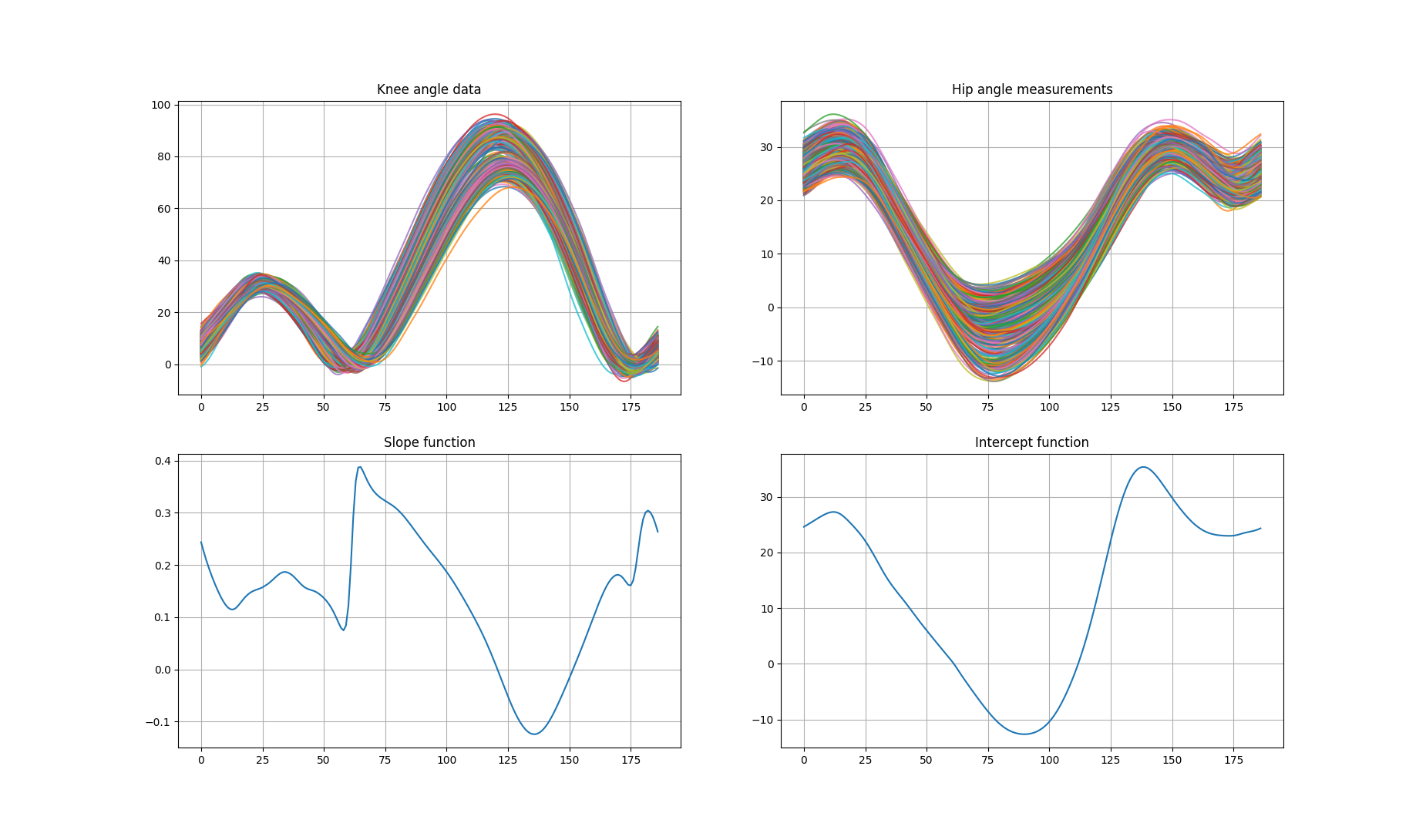}
        \caption{ \textit{Top panels:} knee angle data (left) and hip angle data (right) from a short run where fatigue does not set in. 
\textit{Bottom panels:} concurrent slope $\gamma_i(t)$ (left) and intercept $\alpha_i(t)$ (right) functions under a no-change setting.        
        Collectively, these plots visually demonstrate the coupling structure between the lower-extremity joint angles within the gait cycle for a specific runner.}
        \label{fig:SlopeInterceptNull}
\end{figure}

\section{Mathematical framework}

To formalize our testing framework, let $\alpha(z,t)$ and $\gamma(z,t)$ with $z,t \in [0,1]$ be the concurrent intercept and slope functions in rescaled time such that
$$
\alpha_i(t)=\alpha(i/n,t),\qquad \gamma_i(t)=\gamma(i/n,t), \qquad i=1, \ldots, n.
$$
The functional concurrent linear regression model \eqref{flr_intro} can then be written as
$$
    Y_i (t) = \alpha(i/n,t) + \gamma(i/n,t) X_i (t) + \epsilon_i (t), \quad \quad \E(\epsilon_i(t)|X_i) = 0, \quad \quad i = 1, \dots, n, \ t \in [0,1]. 
$$
We assume that $\alpha(z,t)$ and $\gamma(z,t)$ are continuous in $t\in[0,1]$ and piecewise Lipschitz in $z\in[0,1]$ uniformly over $t$, which allows for single and multiple breaks as well as gradual changes in the regression coefficients. 
We test the joint stability null hypothesis
\begin{align}
    \mathcal H:\quad \alpha(z,t)=\alpha_0(t)\ \text{and}\ \gamma(z,t)=\gamma_0(t),\qquad z,t\in[0,1]. \label{testproblem}
\end{align}
Under the alternative we allow both coefficients to vary, but our test is designed to have power against slope instability, that is,
$$
\mathcal K:\quad \exists\, z_1,z_2\in[0,1]\ \text{such that}\ \sup_{t \in [0,1]} |\gamma(z_1,t)-\gamma(z_2,t)| >0.
$$
In particular, $\mathcal K$ covers the classical single change point alternative $\gamma_i(t)=\gamma_0(t)+\gamma^*(t)\,\mathbbm 1\{i>k^*\}$ for some $k^*\in\{1,\ldots,n-1\}$ and some nonzero $\gamma^*\in C[0,1]$.
We allow $\alpha(z,t)$ to vary arbitrarily under $\mathcal K$, which accommodates intercept shifts that may be induced by changes in $\gamma(z,t)$ whenever $\E(X_i(t))\neq 0$.

Under the null hypothesis, the standard pointwise OLS estimates for the parameters $\gamma_0$ and $\alpha_0$ of the regression model for a sample $(X_1, Y_1), \dots, (X_n, Y_n)$ may be obtained as
\begin{align}
        \hat \gamma (t) &= \bigg(\sum_{i=1}^n ( X_i(t)- \hat\mu_X (t))^2 \bigg)^{-1} \bigg(\sum_{i=1}^n  \big(X_i (t)- \hat\mu_X (t) \big) \big( Y_i (t)- \hat\mu_Y (t)\big) \bigg), \label{slope_est}\\
    \hat \alpha (t) &  = \hat \mu_Y (t) - \hat \gamma (t) \hat \mu_X(t) \label{intercept_est}, \qquad \hat \mu_Y (t) = \frac{1}{n} \sum_{i=1}^n Y_i(t), \quad \hat \mu_X (t) = \frac{1}{n} \sum_{i=1}^n X_i(t).
\end{align}
Corresponding functional OLS residuals are given by
\begin{align}
    \hat \epsilon_i(t) = Y_i(t) - \hat \gamma (t) X_i(t) - \hat \alpha(t). \label{OLS_errors}
\end{align}
We propose the CUSUM-type test statistics using the OLS residuals as follows, 
\begin{align}
          \Q_n  (z, t) & = \frac{1}{\sqrt{n}} \sum_{i=1}^{\floor {nz}} \bigg ( X_i (t)- \hat\mu_X (t) \bigg ) \hat{\epsilon}_i (t), \quad \quad z, t \in [0,1] \label{cusum_slope}.
\end{align}

Note that cumulative sums of the OLS residuals $\hat\epsilon_i(t)$ alone can yield tests with no power in certain settings, for instance when the regressor is centered (see \citealt{jiang2019power} and \citealt{otto2022backward}).
To avoid these issues, we consider the modified CUSUM process \eqref{cusum_slope} using the weighted cumulative sums of $(X_i (t) - \hat \mu_X (t)) \hat \epsilon_i (t)$ to obtain tests that have power against generic breaks in $\gamma(z,t)$.
We propose two different test statistics:
\begin{align}
    \Q_{\text{sup}}  &:= \sup_{z\in[0,1]} \sup_{t\in[0,1]} \big| \Q_n  (z, t) \big|,  \label{q2_supRej002}\\
    	\Q_{L^2}  &:= \sup_{z\in[0,1]} \bigg( \int_0^1 \Q_n (z, t)^2 \,dt \bigg)^{1/2}. \label{q2_L2_Rej}
\end{align}
The decision rule of our test is that we reject $\mathcal H$ at significance level $\rho$ if $\Q_{\text{sup}} \ge q_{1-\rho,\text{sup}}$ or $\Q_{L^2} \ge q_{1-\rho,L^2}$, where $q_{1-\rho,\text{sup}}$ and $q_{1-\rho,L^2}$ denote the $(1-\rho)$ quantiles of the corresponding limiting null distributions.
The specific strengths and hence the power properties of each statistic are discussed later in Section \ref{sec:theoryContributions}. The choice of norm depends on the type of change of interest.
Finally, the CUSUM maximizers in each case of \eqref{q2_supRej002} and \eqref{q2_L2_Rej} are given by
$$
    \hat k_{\text{sup}} = \argmax_{1 \leq i \leq n} \sup_{t\in[0,1]} \big| \Q_n  (i/n, t) \big|, \qquad
    \hat k_{L^2} = \argmax_{1 \leq i \leq n} \bigg( \int_0^1 \Q_n (i/n, t)^2 \,dt \bigg)^{1/2}.
$$

\section{Theoretical contributions}
\label{sec:theoryContributions}

We develop our asymptotic analysis under the following set of assumptions:

\begin{itemize}
\item[(A1)] 
(Exogeneity) The functional time series $(X_i)_{i \in \mathbb Z}$ and $(\epsilon_i)_{i \in \mathbb Z}$ take values in $C[0,1]$ and satisfy $\E(\epsilon_i(t)|X_i) = 0$ for all $t \in [0,1]$, $i \in \mathbb Z$, and $\inf_{t \in [0,1]} \mathrm{Var}(X_i(t)) > 0$.
\item[(A2)] (Moment Bounds) For some Hölder exponent $\eta \in (0,1]$ and moment $p > \max\{4,2/\eta\}$,
$$
\sup_{i\in \mathbb Z} \E([X_i]_\eta^p) < \infty,\quad \sup_{i \in \mathbb Z} \E([\epsilon_i]_\eta^p) < \infty,\quad \sup_{i \in \mathbb Z} \E(\|X_i\|_\infty^p) < \infty,\quad \sup_{i \in \mathbb Z} \E(\|\epsilon_i\|_\infty^p) < \infty,
$$
where $\|f\|_\infty := \sup_{t \in [0,1]} |f(t)|$ is the sup-norm and $[f]_\eta:=\sup_{s\ne t}|f(s)-f(t)|/|s-t|^\eta$ denotes the Hölder seminorm.

\item[(A3)] (Weak Dependence) The combined functional time series $(X_i, \epsilon_i)_{i \in \mathbb Z}$ is weakly stationary and $\beta$-mixing (absolutely regular) of size $-p/(p-4)$ if $\eta > 1/2$ and of size $-p(1-\eta)/(p\eta-2)$ if $\eta \leq 1/2$.
\end{itemize}
\noindent

Condition (A1) is a standard regression model assumption. (A2) imposes different forms of moment bounds. Continuity alone is not sufficient for inference in the sup-norm. We therefore control the modulus of continuity via Hölder conditions. 
For stochastic equicontinuity of sup-norm based estimators, it suffices to assume appropriate moment bounds on the Hölder seminorm and the sup-norm. Here, a finite $p$-th moment of the sup-norm of a $C[0,1]$-valued random variable yields tail control for the overall amplitude and a finite $p$-th moment of the Hölder seminorm gives tail control in the fine-scale oscillations.
(A3) controls the time-dependence of the functional time series. In particular the $\beta$-mixing condition ensures that autocovariances are absolutely summable and that the following long-run covariance kernel exists and is positive semidefinite,
\begin{align}
        C(s, t) = \sum_{l= -\infty}^\infty \mathrm{Cov} \Big((X_0(s)- \mu_X(s))\epsilon_0 (s), (X_l(t) - \mu_X(t))\epsilon_l (t)\Big), \quad \mu_X(t) = \E(X_0(t)). \label{longrunCovariances}
\end{align}
Following \cite{dehling1983limit}, $\beta$-mixing for Banach space-valued processes yields a strong invariance principle with approximation errors of order $O(k^{1/2-\kappa})$ for some $\kappa > 0$, whereas $\alpha$-mixing yields only the slower rate $o(\sqrt{k\log\log k})$. The latter rate is too slow for our statistical applications. While $\beta$-mixing is more restrictive than $\alpha$-mixing, it accommodates Gaussian autoregressive processes (see \citealt{carrasco2002}, \citealt{lu2022}), unlike $\phi$-mixing which excludes such processes, and thus offers a well-suited weak dependence framework for functional time series analysis.
The following theorem adapts the general results of \cite{dehling1983limit} to our specific setting.

\begin{theorem}\label{thm:null_limit}
Under Assumptions (A1)--(A3), there exists a probability space on which we can redefine
$(X_i,\epsilon_i)_{i\in\mathbb Z}$ together with a $C[0,1]$-valued Brownian motion $B(z,t)$ with
$$
\mathrm{Cov}(B(z,s),B(w,t))=\min\{z,w\}\,C(s,t),
$$
where $C(s, t)$ is given in \eqref{longrunCovariances}. Then under the null hypothesis $\mathcal H$, it holds that
$$
\sup_{z,t\in[0,1]}\Big|\Q_n(z,t)-\big(B(z,t)-zB(1,t)\big)\Big|=o_P(1).
$$
\end{theorem}

A direct consequence of Theorem \ref{thm:null_limit} together with the continuous mapping theorem is that the limiting distributions of our test statistics under the null hypothesis $\mathcal H$ are characterized as functionals of $C[0,1]$-valued Brownian bridges as follows,
\begin{align}
	\Q_{\text{sup}} &\overset{d}{\to} \sup_{z,t\in[0,1]} \big| B (z,t) - z B (1,t) \big|, \label{d2_supRej}\\
	\Q_{L^2}  &\overset{d}{\to} \sup_{z\in[0,1]} \bigg( \int_0^1 \big( B (z,t) - z B (1,t) \big)^2 \,dt \bigg)^{1/2}  \label{d2_L2Rej}.
\end{align}
Accordingly, for a given significance level $\rho\in(0,1)$, we reject $\mathcal H$ based on $\Q_{\text{sup}}$ if $\Q_{\text{sup}} > q_{1-\rho,\text{sup}}$, and based on $\Q_{L^2}$ if $\Q_{L^2} > q_{1-\rho,L^2}$, where $q_{1-\rho,\text{sup}}$ and $q_{1-\rho,L^2}$ denote the $(1-\rho)$-quantiles of the respective limiting null distributions in \eqref{d2_supRej}--\eqref{d2_L2Rej}. Their practical determination is discussed in Section \ref{sec:practicalImplementation}.

\begin{theorem}\label{thm:alt_limit}
Suppose Assumptions (A1)--(A3) hold. Under the alternative $\mathcal K$,
$$
\sup_{z,t\in[0,1]}
\Big|
\Q_n(z,t)-\sqrt n\,\sigma_X^2(t)\int_0^z\gamma^*(u,t)\,du
\Big|=O_P(1),
$$
where $\sigma_X^2(t)=\mathrm{Var}(X_0(t))$ and
$\gamma^*(z,t)=\gamma(z,t)-\int_0^1\gamma(u,t)\,du$.
\end{theorem}

Theorem \ref{thm:alt_limit} implies that our tests are consistent whenever there exist two rescaled time points $z_1, z_2 \in[0,1]$ such that $\|\gamma(z_1,\cdot)-\gamma(z_2,\cdot)\|_\infty>0$.
In that case, the mapping $z\mapsto \int_0^z \gamma^*(u,t_0)\,du$
is not identically zero. Together with $\inf_{t\in[0,1]}\sigma_X^2(t)>0$ from (A1), this yields
$$
\sup_{z,t\in[0,1]}\Big|\sigma_X^2(t)\int_0^z \gamma^*(u,t)\,du\Big|>0,
$$
and similarly
$$
\sup_{z\in[0,1]}\int_0^1\Big|\sigma_X^2(t)\int_0^z \gamma^*(u,t)\,du\Big|^2dt>0,
$$
since the remainder in Theorem \ref{thm:alt_limit} is $O_P(1)$ uniformly over $(z,t)$, the $\sqrt n$-drift dominates,
and both $\Q_{\text{sup}}$ and $\Q_{L^2}$ diverge in probability.

\begin{corollary}\label{cor:consistency}
Under the assumptions of Theorem~\ref{thm:alt_limit} and alternative $\mathcal K$, as $n\to\infty$,
$$
	\Q_{\sup}\overset{p}{\to}\infty, \quad \text{and} \quad \Q_{L^2}\overset{p}{\to}\infty.
$$
\end{corollary}

\section{Practical implementation}
\label{sec:practicalImplementation}
\subsection{Estimating the long-run covariance}
According to the main theorem, the asymptotic distribution of the test statistics is correctly characterized only when the relevant long-run covariances in \eqref{longrunCovariances} are available. These are given by
$$
    C (s,t) = c_0 (s,t) + \sum_{l=1}^\infty \Big( c_l (s,t) + c_l (t,s) \Big) \quad \quad \quad \quad s, t \in [0,1],
$$
where each term $c_l (s,t)$ is given by,  
$$
    \qquad c_l(s,t) = \operatorname{Cov}\Big(\big (X_0(s)- \mu_X (s)\big ) \epsilon_0(s), \big (X_{l}(t)- \mu_X(t) \big)\epsilon_{l}(t) \Big),
$$
represents the lag-$l$ covariance. The availability of the long-run covariance $C(s,t)$ is a long-standing problem in functional data analysis and has been discussed previously for example in Section 3.3 of \cite{boniece2023changepoint}.  Additionally, in \cite{rice2017plug} the estimation of the long-run covariance $C(s, t)$ is extensively discussed in the context of bandwidth selection for stationary functional time series. 
Following \cite{rice2017plug}, we estimate the long-run covariance kernel $C(s, t)$ as follows,
\begin{align}
    \widehat C  (s,t)= \hat c_0  (s,t) + \sum_{l = 1}^\infty W_q\big(l/h\big) \Big( \hat c_l (s,t) + \hat c_l (t,s) \Big)
    \label{lrcf}
\end{align}
and the autocovariances are estimated as 
$$
   \hat c_l (s,t) = \frac{1}{n} \sum_{j=1}^{n-l} \big ( X_j(s)  - \hat\mu_X (s) \big ) \hat \epsilon_j(s) \big( X_{j+l}(t) - \hat \mu_X (t) \big )  \hat \epsilon_{j+l}(t), \qquad l \geq 0.
$$
In \eqref{lrcf}, $W_q (\cdot)$ is a symmetric and continuous weight function with bounded support of order $q$ as defined in equations (2)--(3) of \cite{rice2017plug}, and $h$ is a bandwidth parameter. The choice of the weight function depends on user-specific requirements. In adherence to the results and arguments in Section 3 of \cite{bastian2024multiplechangepointdetection} we choose the Quadratic-Spectral weight function for the simulations and the data example used in the rest of this work. 

\subsection{Estimation of quantiles}
\label{est-quantiles}
Let $(\lambda_l , \phi_l)_{l \in \mathbb N}$ denote the orthonormal eigenpairs of the long run covariance kernel $C (s,t)$. Then, by Mercer's theorem, we have the following,
$$
  C (s,t) = \sum_{l=1}^\infty \lambda_l  \phi_l (s) \phi_l (t).
$$
Note that for the $C[0,1]$-valued Brownian motion $B(z,t)$, the Karhunen-Loève expansion in the second component of $B (z,t)$ gives,
$$
    B (z,t) = \sum_{l=1}^\infty \sqrt{\lambda_l } B_l(z) \phi_l(t),
$$
where $B_l(z) = (\lambda_l)^{-1/2}\int_0^1 B(z,t) \phi_l(t) \, dt$. Note that $B_l$ is a Brownian motion with the covariance function,
\begin{align*}
    \text{Cov} (B_l(z_1), B_l(z_2)) 
    &= \lambda_l^{-1} \int_0^1 \int_0^1 \E(B (z_1,s)B (z_2,t)) \phi_l(s) \phi_l(t) \, ds \, dt
    = \min\{z_1, z_2\}. 
\end{align*}

$\{B_l\}$ forms an independent sequence of standard Brownian motions because they are defined as projections from a joint Gaussian process and satisfy $\text{Cov} (B_l(z_1), B_m(z_2)) = 0$ for $l \neq m$. This means that we can represent the limiting Brownian bridge process as,
$$
	B (z,t) - z B (1,t) = \sum_{l=1}^\infty \sqrt{\lambda_l} \phi_l(t) (B_l(z) - z B_l(1)).
$$
 Consequently, we have from our main theorem that,
\begin{align}
	\Q_{\text{sup}} &\overset{d}{\to} \sup_{z,t\in[0,1]} \big| B(z,t) - z B(1,t) \big| \overset{d}{=}  \sup_{z,t\in[0,1]} \Big| \sum_{l=1}^\infty \sqrt{\lambda_l} \phi_l(t) (B_l(z) - z B_l(1)) \Big| \label{asymp_mercerSUP}\\
	\Q_{L^2} &\overset{d}{\to} \sup_{z\in[0,1]} \bigg( \int_0^1 \big( B(z,t) - z B(1,t) \big)^2 \,dt \bigg)^{1/2} \overset{d}{=}  \sup_{z\in[0,1]} \Big( \sum_{l=1}^\infty \lambda_l  (B_l(z)- z B_l(1))^2 \Big)^{1/2},\label{asymp_mercerL2}
\end{align}
where we recall the definitions of $\Q_{\text{sup}}$ and $\Q_{L^2}$ in \eqref{q2_supRej002} and \eqref{q2_L2_Rej}. This representation of the limiting distribution can be estimated and truncated to get the quantile for the decision rule in \eqref{q2_supRej002} and \eqref{q2_L2_Rej} respectively for each test statistic. Note that the summations in \eqref{asymp_mercerSUP} and \eqref{asymp_mercerL2} are truncated to a sum of $m$-terms. Further, as higher order eigenvalues of Hilbert-Schmidt operators are only estimated with lower precision, careful choice of the trimming parameter $m$ is important in practice. We follow the heuristic guideline of \cite{horvath2014testing} and \cite{kokoszka2017testing} where the suggestion is to choose $m$ so that $85\%$ of the variance is explained.   

\begin{singlespace}
\begin{algorithm}[ht]
\begin{algorithmic}[1]
\State \textbf{Given} $\{(Y_i(t),X_i(t)),\, t \in [0,1]\}_{i=1}^n$, metric $\in\{\sup,L^2\}$, significance level $\rho$, Monte Carlo size $R$ (e.g., $R=1000$).
\State \textbf{Compute} $\hat\gamma(t)$, $\hat\alpha(t)$, and residuals $\hat\epsilon_i(t)$ as in equations \eqref{slope_est}--\eqref{OLS_errors}.
\State \textbf{Compute} the test statistic $\Q_{\sup}$ or $\Q_{L^2}$ as in equations \eqref{q2_supRej002}--\eqref{q2_L2_Rej} using $\hat\epsilon_i(t)$.
\State \textbf{Estimate} $\widehat C(s,t)$ as in \eqref{lrcf} and \textbf{compute} eigenpairs $(\hat\lambda_\ell,\hat\phi_\ell)_{\ell\ge1}$ (sorted $\hat\lambda_1\ge\hat\lambda_2\ge\cdots$)
\State \textbf{Set} $m=\min\{M:\sum_{\ell=1}^M \hat\lambda_\ell \ge 0.85\,\mathrm{tr}(\widehat C)\}$, where $\mathrm{tr}(\widehat C)=\int_0^1 \widehat C(t,t)\,dt$.

\For{$r=1,\ldots,R$}
  \State \textbf{Generate} i.i.d.\ standard Brownian bridges
  $\mathbb B^{(r)}_\ell(z) := B^{(r)}_\ell(z)-zB^{(r)}_\ell(1)$, $\ell=1,\ldots,m$, where $B^{(r)}_\ell$ are standard Brownian motions.
  \If{metric = $\sup$}
    \State \textbf{Compute} $T_{\sup}^{(r)} := \sup_{z,t\in[0,1]}
    \Big| \sum_{\ell=1}^m \sqrt{\hat \lambda_\ell}\,\hat \phi_\ell(t)\,\mathbb B^{(r)}_\ell(z) \Big|$.
  \Else
    \State \textbf{Compute} $T_{L^2}^{(r)} := \sup_{z \in [0,1]}
    \Big(\sum_{\ell=1}^m \hat\lambda_\ell\, (\mathbb B^{(r)}_\ell(z))^2\Big)^{1/2}$.
  \EndIf
\EndFor

\If{metric = $\sup$}
  \State \textbf{Set} $q_{1-\rho,\sup}$ as the empirical $(1-\rho)$-quantile of $\{T_{\sup}^{(r)}\}_{r=1}^R$.
  \State \textbf{Reject} $\mathcal H$ if $\Q_{\sup}>q_{1-\rho,\sup}$.
\Else
  \State \textbf{Set} $q_{1-\rho,L^2}$ as the empirical $(1-\rho)$-quantile of $\{T_{L^2}^{(r)}\}_{r=1}^R$.
  \State \textbf{Reject} $\mathcal H$ if $\Q_{L^2}>q_{1-\rho,L^2}$.
\EndIf
\end{algorithmic}
\caption{Change point testing with sup- or $L^2$-norm metric}
\label{basuAlg001}
\end{algorithm}
\end{singlespace}

\section{Monte Carlo simulations}
\label{sec:mc}

This section investigates finite-sample size and power of the proposed tests based on $\Q_{\sup}$ and $\Q_{L^2}$.
We consider both i.i.d.\ and serially dependent functional regressors, and we compare power against (i) global changes of the slope function and (ii) localized spike-like changes.

\begin{figure}[t]
\centering
\includegraphics[width=\linewidth]{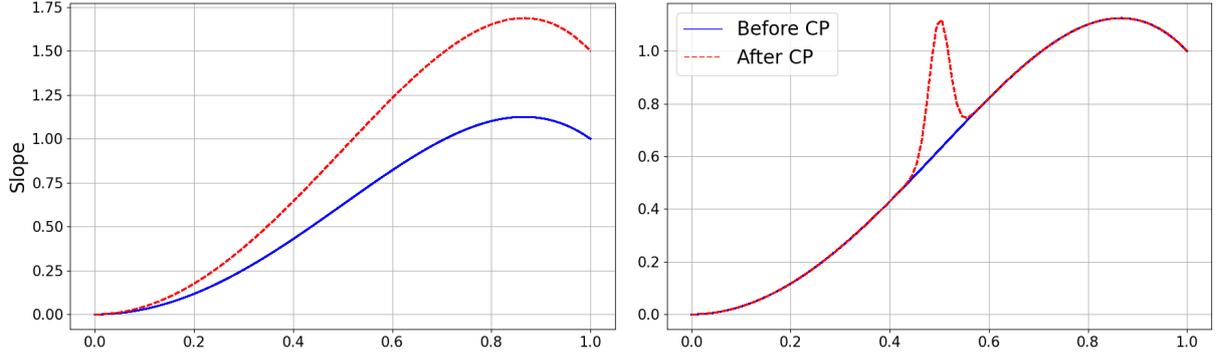}
\caption{\textit{Left:} Scaled alternative \eqref{scaled_transformation}. \textit{Right:} Spiked alternative \eqref{spiked_transform}.
The pre-change slope $\gamma_0$ is shown as a solid curve and the post-change slope $\gamma_1$ as a dashed curve.}
\label{fig:slope_transformation}
\end{figure}

The regressor and error processes are generated using truncated Fourier expansions with truncation
level $D=12$. Let $\{\phi_l\}_{l=1}^D$ and $\{\psi_l\}_{l=1}^D$ denote Fourier bases on $[0,1]$.
For the regressor we use low-frequency functions
$$
\phi_1(t)=\sin(2\pi t),\quad \phi_2(t)=\cos(2\pi t),\quad
\phi_3(t)=\sin(4\pi t),\quad \phi_4(t)=\cos(4\pi t),\ \dots,
$$
whereas for the errors we use higher-frequency components
$$
\psi_1(t)=\sin(12\pi t),\quad \psi_2(t)=\cos(12\pi t),\quad
\psi_3(t)=\sin(24\pi t),\quad \psi_4(t)=\cos(24\pi t),\ \dots .
$$

In the i.i.d.\ setting we generate independent coefficients $N_{i,l}\sim\mathcal N(0,4)$.
To introduce serial dependence, we generate coefficient sequences $\{F_{i,\ell}\}_{i\in\mathbb Z}$ via
\begin{align}
F_{i,l} = 0.8\,F_{i-1,l} + \sigma_l \xi_{i,l},
\qquad
\sigma_l^2 = \frac{\sigma^2}{l^2},
\qquad
\xi_{i,l}\overset{i.i.d.}{\sim}\mathcal N(0,1),
\label{AR_coeffs}
\end{align}
with $\sigma=4$, yielding decreasing variance at higher frequencies.
To simulate a random starting value for the autoregression \eqref{AR_coeffs}, we generate 200 burn-in samples before collecting the sample used in the study.
We set a smooth nonzero mean function
$$
m_X(t)= 9\exp\!\big\{-100\,(t-0.5)^2\big\}, \qquad t\in[0,1],
$$
and define
\begin{align}
X_i(t) &= m_X(t) + \sum_{\ell=1}^D N_{i,\ell}\,\phi_\ell(t),
\qquad\text{(IID-case)}, \label{IIDX}\\
X_i(t) &= m_X(t) + \sum_{\ell=1}^D F_{i,\ell}\,\phi_\ell(t),
\qquad\text{(AR-case)}. \label{ARX}
\end{align}
Regression errors are generated independently across $i$ via
$$
\epsilon_i(t) = \sum_{\ell=1}^D \tilde N_{i,\ell}\,\psi_\ell(t),
\qquad
\tilde N_{i,\ell}\overset{i.i.d.}{\sim}\mathcal N(0,1),
$$
independently of the regressor coefficients.
We use the concurrent regression model \eqref{flr_intro}
with pre-change intercept and slope
\begin{align}
\gamma_0(t) &= t^2(3-2t^3), \label{slope_knot}\\
\alpha_0(t) &= (1-t)^2 + (3-2(1-t)). \label{intercept_knot}
\end{align}
A single change point is introduced at $k^*=\lfloor 0.5\,n\rfloor$ via
$$
\gamma_i(t)=
\begin{cases}
\gamma_0(t), & i<k^*,\\
\gamma_1(t), & i\ge k^*.
\end{cases}
$$
Equivalently, $\gamma_i(t)=\gamma_0(t)+\Delta\gamma(t)\,\mathbbm 1\{i\ge k^*\}$ with
$\Delta\gamma(t) := \gamma_1(t)-\gamma_0(t)$.
We consider two types of post-change slopes $\gamma_1$.
First, we simulate alternatives with a scaled transformation (global change) as follows,
\begin{align}
\gamma_1(t) = \delta\,\gamma_0(t),
\qquad \delta\in(0,1),
\label{scaled_transformation}
\end{align}
so that $\Delta\gamma(t)=(\delta-1)\gamma_0(t)$. In the experiments below we set $\delta=0.5$.
Second, we consider alternatives with spiked transformations (localized changes) using 
\begin{align}
\gamma_1(t) = \gamma_0(t) + \mathcal S(t),
\label{spiked_transform}
\end{align}
where the spike is
\begin{align}
\mathcal S(t)=\frac12\exp\!\left(-\frac12\left(\frac{t-0.5}{0.02}\right)^2\right).
\label{spike_sup}
\end{align}
The two post-change slopes are displayed in Figure \ref{fig:slope_transformation}.

We consider sample sizes $n\in\{100,300,500,1000\}$ under both the IID and AR(1) regressor designs
\eqref{IIDX}--\eqref{ARX}. Tests are carried out at significance level $\rho=0.05$.
For each simulated dataset we compute $\Q_{\sup}$ and $\Q_{L^2}$ and approximate the corresponding
critical values using Algorithm~\ref{basuAlg001} with Monte Carlo size $R=100$ and the same long-run covariance and truncation choices as in Section \ref{sec:practicalImplementation}.
Empirical rejection probabilities are based on $1000$ independent repetitions.

\begin{table}[htbp]
\centering
\begin{tabular}{|c|c||c|c||c|c|}
    \hline
    \multirow{2}{*}{Sample size $(n)$} & \multirow{2}{*}{Setting}
    & \multicolumn{2}{c||}{IID} & \multicolumn{2}{c|}{AR(1)} \\
    \cline{3-6}
    & & $\Q_{L^2}$ & $\Q_{\sup}$ & $\Q_{L^2}$ & $\Q_{\sup}$\\
    \hline \hline
    \multirow{3}{*}{100}
    & $\mathcal H$ & 0.04 & 0.045 & 0.01 & 0.05 \\
    & $\delta=0.5$ & 1.00 & 0.72  & 1.00 & 0.95 \\
    & spike $\mathcal S$ & 0.02 & 0.054 & 0.05 & 0.24 \\
    \hline
    \multirow{3}{*}{300}
    & $\mathcal H$ & 0.05 & 0.03 & 0.05 & 0.05 \\
    & $\delta=0.5$ & 1.00 & 0.99 & 1.00 & 0.98 \\
    & spike $\mathcal S$ & 0.09 & 0.12 & 0.52 & 0.96 \\
    \hline
    \multirow{3}{*}{500}
    & $\mathcal H$ & 0.05 & 0.05 & 0.045 & 0.04 \\
    & $\delta=0.5$ & 1.00 & 1.00 & 1.00 & 0.99 \\
    & spike $\mathcal S$ & 0.12 & 0.20 & 0.90 & 1.00 \\
    \hline
    \multirow{3}{*}{1000}
    & $\mathcal H$ & 0.05 & 0.05 & 0.05 & 0.05 \\
    & $\delta=0.5$ & 1.00 & 1.00 & 1.00 & 1.00 \\
    & spike $\mathcal S$ & 0.19 & 0.89 & 1.00 & 1.00 \\
    \hline
\end{tabular}
\caption{Empirical rejection probabilities at level $\rho=0.05$ for $\Q_{L^2}$ and $\Q_{\sup}$ under IID and AR(1) regressor designs \eqref{IIDX}--\eqref{ARX}.
The change point is $k^*=\lfloor 0.5\,n\rfloor$. The scaled alternative uses $\delta=0.5$ in \eqref{scaled_transformation} and the spiked alternative is given by \eqref{spiked_transform}--\eqref{spike_sup}.}
\label{tab:Q2_only_Jump_Spike}
\end{table}

Table~\ref{tab:Q2_only_Jump_Spike} shows that under the null hypothesis of no change in the slope ($\mathcal H$),
the empirical rejection probabilities are close to and generally below the nominal level $\rho=0.05$.
Under the scaled transformation \eqref{scaled_transformation}, both tests exhibit high power across all sample sizes,
with $\Q_{L^2}$ performing slightly better in detecting this global change.
For the spiked transformation \eqref{spiked_transform}--\eqref{spike_sup}, we observe better performance of $\Q_{\sup}$,
which we attribute to the ability of the sup-norm to detect localized spike-like deviations in functional data.

\section{Application to hip and knee angle data}

\label{sports-data-section}

In this section, we apply our methodology to biomechanical hip and knee angle data obtained from runners. Note that the data used in this work were collected via body-worn sensors, Xsens (Xsens MVN link sensors, sampling at 240 Hz, see \citealt{schepers2018xsens}). Runners, who were not professional athletes and aged 26--39 years were recruited by the \textit{Sports, Data, and Interaction\footnote{\url{http://www.sports-data-interaction.com}}} team and a fatigue protocol was followed. 
This protocol was designed to induce progressive fatigue during the run so that fatigue-related adaptations could emerge over time. Such adaptations are typically more pronounced in non-professional than in professional runners.

An example of the functional time series consisting of the knee and hip angle data is provided in Figure \ref{fig:snapshotData} which is a snapshot from a part of the run. Noticeably such data collected from runners has repetitive cyclic patterns because running involves repetitive movements. As previously mentioned, we will consider each such cycle of the data as a realization of some underlying random function. In the following, we explore two practical applications: the impact of fatigue on (a) co-motion and (b) symmetricity. In the former, the idea is to explore whether the mutual movement of pairwise combinations of the hip, knee and ankle from the same side of the body undergo change due to fatiguing conditions. This study on co-motion or coordination is motivated based on the observed stability of in-phase coordination patterns in experienced runners as compared to novice runners, whose anti-phase motion patterns and higher joint-level variability suggest earlier or more pronounced structural transitions (see \citealt{mo2019differences}). Such differences in coordination stability suggest that, in our setting, experienced runners may exhibit later or very small structural transitions in the estimated slope relationship, while novice runners may show earlier or more pronounced change points. In the latter case of symmetricity, however, we select the knee joint from both the left and right sides of the body to examine whether the symmetry of movement between the two sides is affected by stressor conditions such as fatigue. Our study is motivated due to known results in biomechanics where inter-limb asymmetries may be exaggerated as a result of fatigue leading to a higher injury risk due to unequal force absorption, see \cite{heil2020influence} and literature therein. A precise mathematical model together with inferential methods to identify such irregularities as a result of fatigue enables the biomechanical practitioner to quantify fatigue-driven deviations and implement targeted interventions to mitigate injury risk.

 For each runner, we omit the first and last $20\%$ of observations. The early portion typically reflects the runner acclimating to the long-run conditions, while the final portion captures the natural slowdown toward the end of the session. We focus on fatigue related changes during the course of the run. As each runner has run a different length of time and hence also has different sample sizes, we provide the change locations as a percentage of the considered dataset of the run. Since each runner completes a different duration and thus yields a different number of observations, we report change point locations as percentages of the portion of the run retained for analysis.

\subsection{Change detection in co-motion per side}

We investigate the co-movement of lower-extremity joint angles on a given side of the body throughout a fatiguing run. Specifically, we aim to assess whether the slope parameter of a functional linear regression model which links, for example, the hip to the knee angle, or the knee to the ankle angle, remains stable over time. We choose the right side of five runners, see Table \ref{tab:co-motion} and model the relationship between the hip and knee angle data as follows, 
$$
    Y_{i,\text{hip}} (t) = \alpha_i (t) + \gamma_i (t) X_{i,\text{knee}} (t) + \epsilon_i (t), \quad \quad i = 1, \dots, n,
$$
where $n$ is the total sample size of the functional data for each individual. We run our algorithm for change point detection and find that changes occur in  the slope function for the runners at different parts  of the fatiguing run. The results of our analysis are recorded in Table \ref{tab:co-motion}. The most experienced runner (namely E) shows no change in both the norms. This is similar for the case of runner A who also does not show any change in the $L^2$-norm albeit with a change in the sup-norm.  The other runners who are less experienced show changes in both the sup- and $L^2$-norm. According to studies as in \cite{stergiou1999asynchrony}, \cite{stergiou2001dynamical} more ``in-phase" or coordinated motion between the hip-knee joints enables better shock absorption during running. A lack of change point in the slope function for experienced runners 
may suggest that greater training exposure is associated with a more stable
hip-knee coordination strategy and therefore remaining robust to fatigue. This is also supported in the interpretations of \cite{mo2019differences}.

Runner D shows an interesting result in that the change in $L^2$-norm is picked up about a 100 cycles later than the change in the sup-norm. We show the plot of the slope function before and after change in Figure \ref{fig:compareL2Sup} and it is seen that the sup-norm detects the change point earlier because the alteration in the slope function is initially sharp and localized. Such an (earlier-) localized change may reflect a rapid spike in co-motion variability at a particular part of the gait phase. The magnitude of this change is not picked up by the $L^2$-norm due to the averaging impact of the integral over the domain of the function. 
The later accumulation of global $L^2$ deviation suggests a more sustained reorganization in the hip-knee motion over the entire cycle as fatigue progresses. In prior literature (see \citealt{koblbauer2014kinematic}) it has been seen that fatigue causes runners - especially novices, to adopt a more forward-leaning posture that redistributes mechanical work from fatigued distal structures to the hip and trunk. This altered whole-body strategy modifies the coordination between the hip and knee joints wherein increased reliance on the hip extensors during propulsion is accompanied by compensatory adjustments at the knee. 
The results of an additional analysis of the right knee as regressor and right ankle as the dependent variable for all the runners are presented in the Supplementary Material.

\begin{table}[ht]
\centering \small
\begin{tabular}{|c|c|c||cc||cc|}
    \hline
    Runner Name & Experience (in years) & Sample size &
    \multicolumn{2}{c||}{$L^2$ norm} &
    \multicolumn{2}{c|}{sup-norm} \\
    \hline
    & & & CP slope & \% & CP slope & \% \\
    \hline
    A & 8  & 716  & -   &-  & 177    & 24.7  \\
    B & 3  & 829  & 187  & 22.56  &  184 & 22.2  \\
    C & 2  & 1099 & 255 &23.2    &255   &  23.2 \\
    D & 6  & 1063 & 368 & 34.61  &  273  &   25.7  \\
    E & 11 & 717  & - & -  & -  &  - \\
    \hline
\end{tabular}
\caption{Comparison of changes in slope under $L^2$ \& sup-norms as percentages of retained portion of the run. The (right) knee angle is modeled as the regressor while the (right) hip angle is modeled as the dependent variable. }
\label{tab:co-motion}
\end{table}

\begin{figure}
%\hspace{-1.6cm}
    \includegraphics[width=1\linewidth]{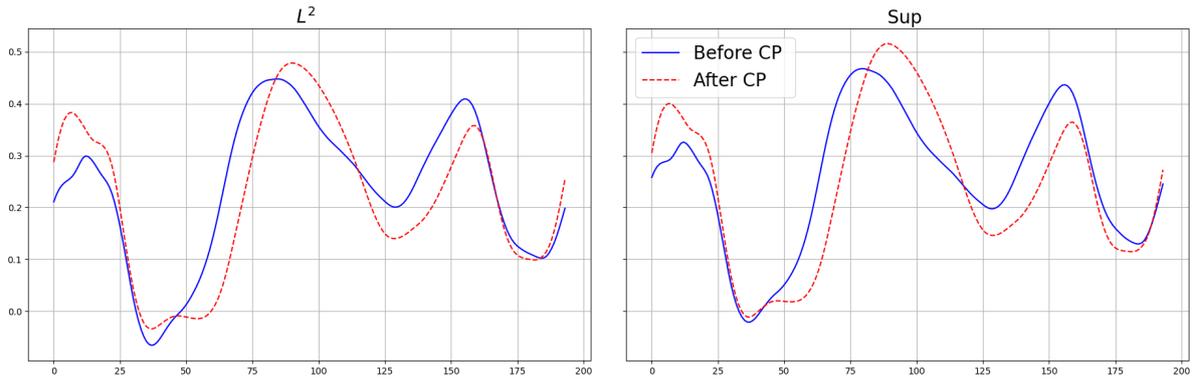}
    \caption{\textit{Left:} Slope comparison before and after the detected change point for Runner D based on the $L^2$-norm. \textit{Right:} Corresponding comparison using the sup-norm. The change point detected under the sup-norm occurs slightly earlier than that obtained with the $L^2$-norm.}
    \label{fig:compareL2Sup}
\end{figure}

\subsection{Symmetricity in joint movement on left and right side}

Under normal conditions of lack of fatigue, our gait movement should remain symmetric through the course of motion. However, this dataset is collected under conditions of fatigue for runners. This means that at some point during the course of the run, the athlete may experience moments of  asymmetry. This structural irregularity as a consequence of fatigue is modeled as a change point detection problem in a linear regression model for the left and right knee joint. We test if this is the case by considering the left knee angle to be the dependent variable and the right knee angle as the regressor variable as follows, 
$$
    Y_{i,\text{Lknee}} (t) = \alpha_i (t) + \gamma_i (t) X_{i,\text{Rknee}} (t) + \epsilon_i (t), \quad \quad i = 1, \dots, n. 
$$
We show that the application of our methods could potentially alert specific runners on aspects to improve when breaks in the slope function i.e., symmetricity of knee joints are reported.  The results of our analysis are recorded in Table \ref{tab:cp_slopeSymmetricity}. Two of the runners (A and E) show no changes with both $L^2$ and sup-norms. As discussed before, these are the most experienced runners. For runners C and D, changes are seen in both norms. Runner B is an exception with change found using $L^2$-norm and no change found using sup-norm. Such inter-limb movement asymmetries, particularly under fatigue, have been the focus of numerous biomechanical investigations (see \citealt{heil2020influence}, \citealt{radzak2017asymmetry}, \citealt{apte2021biomechanical}). Despite this extensive work, the literature consistently notes that key aspects of fatigue-related asymmetry are still not fully understood. \cite{gao2022effects} report that symmetry changes, in particular in the knee joint, due to fatigue may be a reason for running related injuries and running performance while fatigue related asymmetries in lower limbs, particularly in the case of amateur runners has been seen in \cite{gao2020effect}. 
An additional study for symmetricity between the left and right hip angle data for all the runners is presented in the Supplementary Material.

\begin{table}[htbp]
    \centering \small
    \begin{tabular}{|c|c||cc||cc|}
    \hline
       Runner & Sample size &
       \multicolumn{2}{c||}{$L^2$ norm} &
       \multicolumn{2}{c|}{sup-norm} \\
       \hline
       & & CP-slope & \%-run & CP-slope & \%-run \\
       \hline
       A & 716  &   - & -   &  -  & -\\
       B & 829 &   256 & 30.9  & -   & -  \\
       C & 1099  & 689 &  62.7 & 728  & 66.24  \\
       D & 1063 & 306   & 28.79  &  306  &  28.79 \\
       E & 717 &  -  & -  &   -  & -\\
       \hline
    \end{tabular}
    \caption{CP-slope results under $L^2$- and sup-norms testing for change in symmetricity on the left and right knees, with the left knee as the dependent variable. }
    \label{tab:cp_slopeSymmetricity}
\end{table}

\subsection{Varying-coefficient joint coupling}
\label{varyingcoeff-jointcoupling}
\begin{figure}[ht]
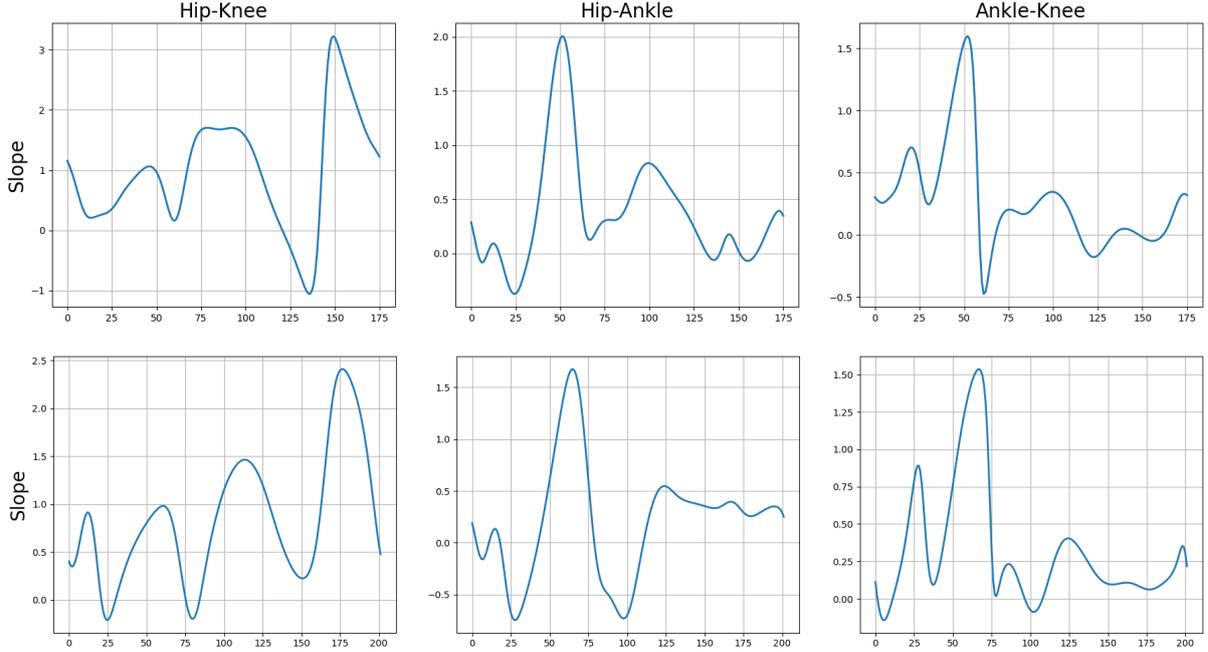

    \centering
    \includegraphics[width=\linewidth]{images/MostExpRunnerE.pdf}
\includegraphics[width =\linewidth]{images/LeastExpRunnerC.pdf}
\caption{\textit{Top row:} Coordination profiles of most experienced athlete (Runner E in this work). \textit{Bottom row: }profiles of least experienced athlete (Runner C). \textit{Both: }Slope functions with Hip-Knee (left panel), Hip-Ankle (middle panel) and Ankle-Knee (right panel) with the latter being the regressor for each combination.}
\label{CoordinationComparison}
\end{figure}

As an additional illustration, we take a step back from the change point inference methodology to the functional concurrent regression model. More precisely, we are interested in interpreting the biomechanical relevance of the parameters of the functional linear regression model, especially the slope function. Mathematically, the slope function characterizes the relationship between the regressor function, say the gait patterns from the knee with respect to the dependent variable, such as gait patterns from the hip. This provides a mathematical basis and quantification for biomechanical concepts like inter-limb coupling. The inter-limb movement during activities like running, walking etc.\ have been previously studied in \cite{davis2019assessing} and \cite{dierks2007discrete}. The coordination between lower extremity joints is interesting irrespective of whether it is fatiguing run or not. It provides a snapshot of each individual's style of movement during an activity. We provide in Figure \ref{CoordinationComparison} the slope function in Hip-Knee, Hip-Ankle and Ankle-Knee regression setups. The latter in each pair is always the regressor. We make a comparison between the extreme experience levels of two runners (C \& E) from a small, non-fatiguing portion of the run.  The application of the algorithm in \cite{bastian2024multiplechangepointdetection} showed that these runners have two change points and we choose the (stable) segment between the first and second changes. The experienced runner shows clear, stable, and well-timed coordination patterns between hip, knee, and ankle, reflected in smooth and distinct slope profiles across the gait cycle. In contrast, the novice runner exhibits more variable, diffuse, and weaker coupling, indicating less consistent inter-joint control. Overall, the profiles highlight how joint coordination may become more structured with increasing running experience. Naturally further intra-group level studies need to be performed in order to show conclusive evidence of such trends. Nonetheless, a slope function gives practitioners a clear, time-specific measure of joint coordination, helping them pinpoint inefficient or compensatory movement patterns and focus training or rehabilitation on the exact phases where coordination breaks down.

\section{Conclusion and outlook}

We developed a change point testing methodology for concurrent functional linear regression models with $C[0,1]$-valued time series. The proposed tests are based on a CUSUM process of regressor-weighted OLS residual functions, which avoids the well-known loss of power that can arise from CUSUM procedures applied to residuals alone when the regressor is centered. To accommodate different types of structural changes in the slope function, we introduced both an $L^2$-based statistic, which is sensitive to global deviations, and a sup-norm statistic, which is particularly suited for localized spike-like changes. Under Hölder regularity and weak dependence conditions, we established a functional strong invariance principle in $C[0,1]$ based on \cite{dehling1983limit}, derived the limiting null distributions as functionals of a $C[0,1]$-valued Brownian bridge, and showed consistency against a broad class of slope-instability alternatives. Practical implementation is enabled via long-run covariance estimation and an eigen-expansion based Monte Carlo procedure to obtain critical values.

Our simulations support these theoretical findings by showing satisfactory size control and highlighting complementary power properties: $\Q_{L^2}$ performs well for diffuse, global slope changes, whereas $\Q_{\sup}$ is more sensitive to localized changes. From a practitioner perspective, this suggests selecting the norm according to the anticipated change mechanism, or reporting both statistics when the nature of potential instability is unknown.

In the sports application, the fitted slope functions provide an interpretable, phase-specific description of inter-joint coupling over the gait cycle, and the change point procedure identifies when these coupling patterns shift during a fatiguing run. Future research in biomechanics could apply these methods to larger cohorts and varied running conditions (e.g., fatigue protocols, speed changes, surface types) to better understand how coordination patterns adapt. Extending the approach to multi-joint kinematics may reveal broader compensatory mechanisms that precede injury. Longitudinal studies could examine how training, rehabilitation, or technique interventions modify slope functions and change point patterns over time, providing actionable tools for coaches, clinicians, and performance scientists.

\section*{Acknowledgements}

This work was supported by the  
Deutsche Forschungsgemeinschaft (DFG) through the project titled: ``\textit{Modeling functional time series with dynamic factor structures and points of impact}'',  with project number 511905296.  Thanks to Patrick Bastian for helpful discussions and the team of \textit{Sports, Data, and Interaction\footnote{\url{http://www.sports-data-interaction.com}}},  in particular Robbert van Middelaar and Aswin Balasubramaniam for meticulously collecting, pre-processing and providing the data used in this work. We furthermore thank the Regional Computing Center of the University of Cologne (RRZK) for providing computing time on the DFG-funded (Funding number: INST 216/512/1FUGG) High Performance Computing (HPC) system CHEOPS.

\newpage

\begin{center}
{\large\bf SUPPLEMENTARY MATERIAL TO} \\ \vspace{2ex}
{\Large \textbf{Detecting Parameter Instabilities in Functional Concurrent Linear Regression}} \\
%{\large by Sven Otto and Nazarii Salish}
\end{center}

\setcounter{equation}{0} \renewcommand{\theequation}{A.\arabic{equation}}
\setcounter{lemma}{0} \renewcommand{\thelemma}{A.\arabic{lemma}}
\appendix

\section{Additional results from the sports data example}
\begin{table}[ht] \small
\centering
\begin{tabular}{|c|c|c||cc||cc|}
    \hline
    Runner Name & Experience (in years) & Sample size &
    \multicolumn{2}{c||}{$L^2$ norm} &
    \multicolumn{2}{c|}{sup-norm} \\
    \hline
    & & & CP slope & \% & CP slope & \% \\
    \hline
    A & 8  & 716   & -    & -&  -   &  -  \\
    B & 3  &  829 & -   & -   &  -  & -  \\
    C & 2  & 1099 &   300 & 27.3   &  300  & 27.3    \\
    D & 6  & 1063  &  212  & 19.9   &  216    & 20.3     \\
    E & 11 &  717  & -  &-& -   &-\\
    \hline
\end{tabular}
\caption{Comparison of changes in slope under $L^2$ \& sup-norms as percentages of retained portion of the run. The (right) knee angle is modelled as the regressor while the (right) ankle angle is modelled as the dependent variable. }
\label{tab:co-motionRKRA}
\end{table}

\begin{table}[ht] \small
\centering
\begin{tabular}{|c|c|c||cc||cc|}
    \hline
    Runner Name & Experience (in years) & Sample size &
    \multicolumn{2}{c||}{$L^2$ norm} &
    \multicolumn{2}{c|}{sup-norm} \\
    \hline
    & & & CP slope & \% & CP slope & \% \\
    \hline
    A & 8  & 716   &   163   &22.77  &    163  &22.77     \\
    B & 3  &  829 &   642  & 77.44    &    643 & 77.56    \\
    C & 2  & 1099 &   246  &   22.38  &  839   &76.34      \\
    D & 6  & 1063  &  270   &  25.4   &    235   &  22.10    \\
    E & 11 &  717  & -  & - & -   & -\\
    \hline
\end{tabular}
\caption{Comparison of changes in slope under $L^2$ \& sup-norms as percentages of retained portion of the run. The (right) hip angle is modelled as the regressor while the (left) hip angle is modelled as the dependent variable. }
\label{tab:symmetricityRHLH}
\end{table}

\section{Technical Proofs}

\subsection{Definitions and Auxiliary Results}

Let $(S,\tau)$ be a compact metric space and let $C(S)$ be the separable Banach space of continuous functions $f: S \to \mathbb R$, equipped with the sup-norm.
A $C(S)$-valued time series $(Z_i)_{i \in \mathbb Z}$ is called weakly stationary if, for all $i, j, k \in \mathbb Z$, and for all points $s,t \in S$ it holds that
$$
	\E\big(Z_i(s) \big) = \E\big(Z_{i+k}(s)\big), \quad \quad \E\big(Z_i(s)Z_j(t)\big) = \E\big(Z_{i+k}(s)Z_{j+k}(t)\big) < \infty.
$$
Define $\mathcal F_{-\infty}^i = \sigma(\ldots, Z_{i-2}, Z_{i-1}, Z_i)$ and $\mathcal F_{i+h}^\infty = \sigma(Z_{i+h}, Z_{i+h+1}, \ldots)$.
The time series $(Z_i)_{i \in \mathbb Z}$ is called $\alpha$-mixing (strong mixing) if 
$$
	\alpha(h) := \sup_{i \in \mathbb Z} \sup_{\substack{A \in \mathcal F_{-\infty}^i \\ B \in \mathcal F_{i+h}^\infty}} \big|P(A \cap B) - P(A)P(B)\big| \overset{h \to \infty}{\longrightarrow} 0,
$$
and $(Z_i)_{i \in \mathbb Z}$ is called $\beta$-mixing (absolutely regular) if 
$$
	\beta(h) := \sup_{i \in \mathbb Z} \E\bigg( \sup_{B \in \mathcal F_{i+h}^\infty} \big|P(B|\mathcal F_{-\infty}^i) - P(B)\big|  \bigg) \overset{h \to \infty}{\longrightarrow} 0.
$$
Since $\alpha(h) \leq \beta(h)$ (see, e.g., \citealt{doukhan1994}, §1.1 Proposition 1), $\beta$-mixing implies $\alpha$-mixing.
We say that $(Z_i)_{i \in \mathbb Z}$ is $\beta$-mixing of size $-\varphi_0$ if $\beta(h)=O(h^{-\varphi})$ for some $\varphi>\varphi_0$, as $h \to \infty$.
The mixing size is preserved under finite-lag measurable transformations of the form $Y_i=g(Z_i,\ldots,Z_{i-l})$ because, for all $h\ge l$,  $\beta_Y(h)\le \beta(h-l)$, where $\beta_Y(h)$ are the $\beta$-mixing coefficients of $(Y_i)_{i \in \mathbb Z}$.
Hence, if $(Z_i)_{i \in \mathbb Z}$ is $\beta$-mixing of size $-\varphi_0$, then $(Y_i)_{i \in \mathbb Z}$ is also.

The relation $\alpha(h) \leq \beta(h)$ implies that many mixing inequalities that hold for $\alpha$-mixing processes also apply to $\beta$-mixing processes. In particular, the following useful inequality from Lemma 3.11 in \cite{dehling2002empirical} carries over directly.

\begin{lemma} \label{lem:mixinginequality}
Let $(Z_i)_{i \in \mathbb Z}$ be a $C(S)$-valued $\beta$-mixing functional time series. 
For any $1 \leq m,p,q < \infty$ satisfying $m^{-1} + p^{-1} + q^{-1} = 1$ we have
\begin{align*}
	\big| \mathrm{Cov}(Z_i(s),  Z_{i+h}(t)) \big| \leq 10 \, \beta(h)^{1/m}\ \E(|Z_i(s)|^p)^{1/p}\ \E(|Z_{i+h}(t)|^q)^{1/q}, \qquad s,t \in S.
\end{align*}
\end{lemma}

\noindent
The following result forms the cornerstone of our asymptotic theory and is an immediate consequence of Theorems 3 and 6 in \cite{dehling1983limit} applied to $C(S)$-valued processes.

\begin{lemma} (\citealt{dehling1983limit})
\label{lem:DehlingIP}
Let $(S,\tau)$ be a compact metric space and $(Z_i)_{i \in \mathbb Z}$ a weakly stationary $C(S)$-valued time series with $\E(Z_i(s)) = 0$ for all $s \in S$ and, for some $q > 2$, $\sup_{i \in \mathbb Z} \E(\sup_{s \in S} |Z_i(s)|^q) < \infty$.
Assume the sequence is $\beta$-mixing of size $-q/(q-2)$.
Let $N_\tau(S,\epsilon)$ be the covering number of $(S,\tau)$, defined as the minimum number of $\epsilon$-balls needed to cover $S$, fix $g(\cdot)$ with $N_\tau(S,g(m))=m$, and suppose
\begin{align} \label{eq:DehlingTightnessCondition}
	\sup_{n \in \mathbb N} \E\Bigg( \sup_{\tau(s,t) \leq g(m)} \Big| n^{-1/2} \sum_{i=1}^{n} (Z_i(s) - Z_i(t)) \Big|^2 \Bigg) = O(m^{-\kappa}) \quad \text{for some} \ \kappa > 0.
\end{align}
Then there is a probability space on which one can redefine $(Z_i)_{i \in \mathbb Z}$ together with a $C(S)$-valued Brownian motion $B(k,s)$, $k \geq 0$, $k \in \mathbb Z$, $s \in S$, such that
\begin{align*}
      \sup_{s \in S}  \bigg|\sum_{i=1}^{\lfloor k \rfloor} Z_i(s)- B(k,s)\bigg| = O(k^{1/2 - \kappa}) \quad a.s., \quad \text{for some} \  \kappa > 0, \quad (k \rightarrow\infty),
\end{align*}
with $\mathrm{Cov}(B(v,s), B(w,t)) = \min\{v,w\} C(s,t)$ and $C(s,t) = \sum_{h=-\infty}^\infty \mathrm{Cov}(Z_0(s), Z_h(t))$. 
\end{lemma}

\noindent
The next result is a Rosenthal-type inequality for $\alpha$-mixing sequences, which is shown in Lemma S.2 of \cite{fan2023} as a consequence of Theorem 6.3 in \cite{rio2017}. The relation $\alpha(h) \leq \beta(h)$ implies that it carries over directly to the $\beta$-mixing coefficient:

\begin{lemma}[\citealt{fan2023}, Lemma S.2] \label{lem:Rosenthalinequality}
Let $r\ge2$ and $q>r$. Let $(x_i)_{i\in\mathbb Z}$ be a real-valued mean-zero time series with $\beta$-mixing coefficients $\beta(h)$. Suppose $\max_{1\le i\le n}\E(|x_i|^{q})<\infty$. Then, for every $n \ge 1$,
$$
\E\bigg(\Big|\sum_{i=1}^{n} x_i\Big|^{r}\bigg)^{1/r} \leq 
C_{r,q} 
\Bigg(\sum_{h=0}^{n-1} (h+1)^{r-2} \beta(h)^{1-r/q}\Bigg)^{1/r}
 \Big(\max_{1\le i\le n} \E(|x_i|^{q})\Big)^{1/q} \sqrt{n},
$$
where the constant $C_{r,q}< \infty$ depends only on $r$ and $q$.
\end{lemma}

\noindent
The following maximal inequality follows as a special case from Theorem 2.2.4 in \cite{vandervaart2023} applied to the Orlicz norm with Young function $\psi(x) = x^r$ and inverse $\psi^{-1}(x) = x^{1/r}$ for any $r \geq 1$:
\begin{lemma}[\citealt{vandervaart2023}, Theorem 2.2.4] \label{lem:vandervaart}
Let $(S,\tau)$ be a compact metric space and let $R$ be a $C(S)$-valued random variable with
\begin{align}\label{eq:vandervaartCondition}
	\E\Big(|R(s) - R(t)|^r\Big)^{1/r} \leq C \tau(s,t), \quad \text{for all} \ s,t\in S \ \text{and some} \ C< \infty, \ r\geq 1.
\end{align}
Then, for any $\nu > 0$ and $\delta > 0$,
$$
	\E\bigg( \sup_{\tau(s,t) \leq \delta} |R(s) - R(t)|^r \bigg)^{1/r} \leq K \left[\int_0^\nu  (D(\epsilon))^{1/r} \, d\epsilon + \delta (D(\nu))^{2/r} \right],
$$
where $K < \infty$ is a constant.
Here, $D(\epsilon)$ denotes the packing number, defined as the maximum number of $\epsilon$-separated points in $(S,\tau)$.
\end{lemma}

\subsection{Strong invariance principle}

\begin{lemma} \label{lem:SIP}
	Let $(Z_i)_{i \in \mathbb Z}$ be a $C[0,1]$-valued weakly stationary functional time series with $\E(Z_i(t)) = 0$ for $t \in [0,1]$ such that, for some Hölder exponent $\eta \in (0,1]$ and moment $q > \max\{2,1/\eta\}$, $\sup_{i \in \mathbb Z } \E([Z_i]_\eta^q) < \infty$ and $\sup_{i \in \mathbb Z} \E(\|Z_i\|_\infty^q) < \infty$.
	Assume $(Z_i)_{i \in \mathbb Z}$ is $\beta$-mixing of size $-q/(q-2)$ if $\eta > 1/2$ and of size $-q(1-\eta)/(q\eta - 1)$ if $\eta \leq 1/2$.
Then, there exists a probability space on which we can redefine $(Z_i)_{i \in \mathbb Z}$ together with a $C[0,1]$-valued Brownian motion $B(k,t)$ with
$$
	\mathrm{Cov}(B(k,s), B(l,t)) = \min\{k,l\} C(s,t), \qquad C(s,t) = \sum_{h=-\infty}^\infty \mathrm{Cov}(Z_0(s), Z_h(t)),
$$
where $k,l \geq 0$, $s,t \in [0,1]$, 
such that for some $\kappa > 0$,
\begin{align*}
&\sup_{t \in [0,1]} \left|\sum_{i=1}^{\lfloor k \rfloor} Z_i(t) - B(k,t)\right| = O(k^{1/2 - \kappa}), \qquad \text{a.s.} \quad (k \to \infty). 
\end{align*}
\end{lemma}

\begin{proof}
First, note that by Lemma \ref{lem:mixinginequality} and the moment and mixing assumptions, the series defining $C(s,t)$ is absolutely convergent for all $s,t\in[0,1]$.
We will apply Lemma \ref{lem:DehlingIP} for $S = [0,1]$ together with the metric $\tau_\eta(s, t) = |s-t|^{\eta}$.
The Hölder moment condition implies that, for some $K < \infty$ and all $s,t \in [0,1]$,
\begin{align} \label{eq:IPtheorem.hölderbound}
\sup_{i \in \mathbb Z} \E\big(|Z_i(s) - Z_i(t)|^{q}\big)^{1/q} \leq \sup_{i \in \mathbb Z} \E\big([Z_i]_\eta^q \big)^{1/q} |s-t|^{\eta} \leq K \tau_\eta(s,t).
\end{align}

It remains to verify condition \eqref{eq:DehlingTightnessCondition}.
For the compact metric space $([0,1], \tau_\eta)$, the covering number is given by $N_{\tau_\eta}(S,\epsilon) = C_N \epsilon^{-1/\eta}$ for some constant $C_N$ and the packing number is $D(\epsilon) = C_D \epsilon^{-1/\eta}$ for some constant $C_D$ (see, e.g., Section 2 of \citealt{vandervaart2023} for a detailed treatment of covering and packing numbers). The inverse function $g$ satisfying $N_{\tau_\eta}(S,g(m))=m$ is given by $g(m) = C_N^\eta m^{-\eta}$.

Consider the partial sum process $R_n(s) := n^{-1/2} \sum_{i=1}^{n} Z_i(s)$ for $s \in [0,1]$.
For any $r \in [2,q)$, Lemma \ref{lem:Rosenthalinequality} and \eqref{eq:IPtheorem.hölderbound} imply that
\begin{align*}
	\E\big(|R_n(s) - R_n(t)|^r\big)^{1/r}  
	&= \frac{1}{\sqrt n} \E\bigg( \Big| \sum_{i=1}^n (Z_i(s) - Z_i(t)) \Big|^r \bigg)^{1/r} \\
	&\leq C_{r,q} \bigg( \sum_{h=0}^{n-1} (h+1)^{r-2} \beta(h)^{1-r/q} \bigg)^{1/r} \Big(\max_{1\le i\le n} \E(|Z_i(s) - Z_i(t)|^{q})\Big)^{1/q} \\
	&\leq C_{r,q} K \tau_\eta(s,t) \bigg( \sum_{h=0}^{n-1} (h+1)^{r-2} \beta(h)^{(q-r)/q} \bigg)^{1/r}.
\end{align*}
For the case $\eta > 1/2$ set $r = 2$. By the assumed mixing size, $\beta(h)=O(h^{-\varphi})$ for some $\varphi>q/(q-2)$, which implies
$$
	\sum_{h=0}^{n-1} (h+1)^{r-2} \beta(h)^{(q-r)/q} \leq \sum_{h=0}^{\infty} \beta(h)^{(q-2)/q} < \infty.
$$
For the case $\eta \leq 1/2$, by the assumed mixing size, there exists $\varphi > q(1-\eta)/(q\eta-1)$ such that $\beta(h) = O(h^{-\varphi})$. Fix any 
$r \in (1/\eta, \min\{q,q(1+\varphi)/(q+\varphi)\})$. The interval is nonempty because $q(1+\varphi)/(q+\varphi) > 1/\eta$ is equivalent to $\varphi > q(1-\eta)/(q\eta-1)$, and $q > 1/\eta$ holds by assumption. The condition $r < q(1+\varphi)/(q+\varphi)$ implies $\varphi(q-r)/q > r-1$. Then, for some $\tilde C <\infty$,
$$
	\sum_{h=0}^{n-1} (h+1)^{r-2} \beta(h)^{(q-r)/q} \leq \tilde C \sum_{h=1}^{\infty} h^{r-2-\varphi(q-r)/q} < \infty.
$$
Hence, condition \eqref{eq:vandervaartCondition} of Lemma \ref{lem:vandervaart} is satisfied for our specific choice of $r$ depending on $\eta$ with the property that $r\eta > 1$.
Therefore, by the monotonicity of $L^r$-norms and Lemma \ref{lem:vandervaart},
when setting $\nu = \delta^{r\eta/(r\eta+1)}$,
\begin{align*}
\sup_{n \in \mathbb N} \E\bigg( \sup_{\tau_\eta(s,t) \leq \delta} |R_n(s) - R_n(t)|^2 \bigg)^{1/2} 
&\leq\sup_{n \in \mathbb N}	\E\bigg( \sup_{\tau_\eta(s,t) \leq \delta} |R_n(s) - R_n(t)|^r \bigg)^{1/r} \\
&\leq  K \left[\int_0^\nu  (D(\epsilon))^{1/r} \, d\epsilon + \delta (D(\nu))^{2/r} \right] \\
&= K C_D^{1/r} \frac{\nu^{1-1/(r\eta)}}{1-1/(r \eta)} + KC_D^{2/r} \delta \nu^{-2/(r\eta)}  \\
&\leq \tilde K \delta^{(r\eta-1)/(r\eta+1)},
\end{align*}
for some $\tilde K < \infty$,
where the integral $\int_0^\nu (D(\epsilon))^{1/r} \, d\epsilon = \int_0^\nu C_D^{1/r} \epsilon^{-1/(r \eta)} \, d\epsilon$ converges because $r\eta > 1$.
Then, with $\delta = g(m) = C_N^\eta m^{-\eta}$, we obtain
\begin{align*}
    \sup_{n \in \mathbb N} \E\Bigg( \sup_{\tau_\eta(s,t) \leq g(m)} \Big| \frac{1}{\sqrt n} \sum_{j=1}^{n} (Z_j(s) - Z_j(t)) \Big|^2 \Bigg) 
\leq \tilde K^2 C_N^{2\eta(r\eta-1)/(r\eta+1)} m^{-2\eta(r\eta-1)/(r\eta+1)},
\end{align*}
so condition \eqref{eq:DehlingTightnessCondition} holds with $\kappa =2\eta(r\eta-1)/(r\eta+1) > 0$, and the assertion follows.
\end{proof}

\subsection{Auxiliary results for Theorems \ref{thm:null_limit} and \ref{thm:alt_limit}}

\subsubsection{Decomposition of $\Q_n$}

Define the centered curves as follows, 
$$
	\hat x_i(t) := X_i(t) - \hat \mu_X(t), \quad \hat y_i(t) := Y_i(t) - \hat \mu_Y(t),
$$
where $\hat \mu_X(t) = n^{-1} \sum_{i=1}^n X_i(t)$ and $\hat \mu_Y(t) = n^{-1} \sum_{i=1}^n Y_i(t)$.
Further with $\hat \alpha(t) = \hat \mu_Y(t) - \hat \gamma(t) \hat \mu_X(t)$ the OLS residuals satisfy, 
$$
	\hat \epsilon_i(t) = Y_i(t) - \hat \alpha(t) - \hat \gamma(t) X_i(t) = \hat y_i(t) - \hat \gamma(t) \hat x_i(t), \quad i=1, \ldots, n.
$$
We introduce the partial sum processes,
$$
	U_n(z,t) := \frac{1}{n}\sum_{i=1}^{\floor{nz}} \hat x_i(t) \hat y_i(t), \quad V_n(z,t) := \frac{1}{n} \sum_{i=1}^{\floor{nz}} (\hat x_i(t))^2.
$$
Then the OLS slope admits the representation given by,
$$
	\hat \gamma(t) = \frac{U_n(1,t)}{V_n(1,t)},
$$
which consequently leads to,
\begin{align}
\Q_n(z,t)
&=\frac{1}{\sqrt n}\sum_{i=1}^{\lfloor nz\rfloor}\big(X_i(t)-\hat\mu_X(t)\big)\hat\epsilon_i(t) \notag \\
&=\frac{1}{\sqrt n}\sum_{i=1}^{\lfloor nz\rfloor}\hat x_i(t)\big(\hat y_i(t)-\hat\gamma(t)\hat x_i(t)\big) \notag \\
&=\sqrt n\Big(U_n(z,t)-\frac{V_n(z,t)}{V_n(1,t)} \, U_n(1,t)\Big), \label{eq:Q-bridgestructure}
\end{align}
which reveals the bridge-type structure of $\Q_n(z,t)$.

\subsubsection{Decomposition of $U_n$}
We define the population-centered variables as follows, 
\begin{align*}
x_i(t)&:=X_i(t)-\mu_X(t),\qquad \mu_X(t):=\E(X_0(t)),\\
y_i(t)&:=Y_i(t)-\bar\mu_Y(t),\qquad \bar\mu_Y(t):=\frac1n\sum_{j=1}^n \E(Y_j(t)),
\end{align*}
and the mean estimation errors as follows, 
$$
\delta_X(t):=\hat\mu_X(t)-\mu_X(t),\qquad \delta_Y(t):=\hat\mu_Y(t)-\bar\mu_Y(t).
$$
Then $\hat x_i(t)=x_i(t)-\delta_X(t)$ and $\hat y_i(t)=y_i(t)-\delta_Y(t)$. Hence,
\begin{align}
U_n(z,t)
=\frac1n\sum_{i=1}^{\lfloor nz\rfloor}\hat x_i(t)\hat y_i(t) 
=\frac1n\sum_{i=1}^{\lfloor nz\rfloor}x_i(t)y_i(t) + S_{\text{cent},n}(z,t), \label{eq:Un_split_xy_H}
\end{align}
where it holds that, 
\begin{align} \label{eq:S-cent}
S_{\text{cent},n}(z,t)
:=-\delta_Y(t)\,\frac1n\sum_{i=1}^{\lfloor nz\rfloor}x_i(t)
-\delta_X(t)\,\frac1n\sum_{i=1}^{\lfloor nz\rfloor}y_i(t)
+\frac{\lfloor nz\rfloor}{n}\,\delta_X(t)\delta_Y(t). 
\end{align}
Moreover, we define the deterministic sample averages as follows,
$$
\bar \alpha_n(t):=\frac1n\sum_{j=1}^n \alpha(j/n,t),\qquad
\bar \gamma_n(t):=\frac1n\sum_{j=1}^n \gamma(j/n,t),
$$
and the centered coefficient functions as,
$$
\alpha_i^*(t):=\alpha(i/n,t)-\bar\alpha_n(t),\qquad
\gamma_i^*(t):=\gamma(i/n,t)-\bar\gamma_n(t).
$$
Using the model $Y_i(t)=\alpha(i/n,t)+\gamma(i/n,t)X_i(t)+\epsilon_i(t)$ and the identity,
$$
\bar\mu_Y(t)=\frac1n\sum_{j=1}^n \E(Y_j(t))=\bar\alpha_n(t)+\mu_X(t)\bar\gamma_n(t),
$$
we obtain the following,
$$
y_i(t)=\alpha_i^*(t)+\gamma(i/n,t)\,x_i(t)+\mu_X(t)\gamma_i^*(t)+\epsilon_i(t),
$$
and therefore it holds that,
$$
x_i(t)y_i(t)=x_i(t)\epsilon_i(t)+\gamma(i/n,t)\,x_i(t)^2+\big(\alpha_i^*(t)+\mu_X(t)\gamma_i^*(t)\big)x_i(t).
$$
Consequently, we have that,
\begin{align}
\frac1n\sum_{i=1}^{\lfloor nz\rfloor}x_i(t)y_i(t)
&=S_{\epsilon,n}(z,t)+S_{\gamma,n}(z,t)+S_{\text{coef},n}(z,t), \label{eq:xy_split}
\end{align}
where it holds that,
\begin{align*}
S_{\epsilon,n}(z,t) &:= \frac{1}{n}\sum_{i=1}^{\lfloor nz\rfloor} x_i(t)\epsilon_i(t),\\
S_{\gamma,n}(z,t) &:= \frac{1}{n}\sum_{i=1}^{\lfloor nz\rfloor} \gamma(i/n,t)\,x_i(t)^2,\\
S_{\text{coef},n}(z,t) &:= \frac{1}{n}\sum_{i=1}^{\lfloor nz\rfloor}\big(\alpha_i^*(t)+\mu_X(t)\gamma_i^*(t)\big)\,x_i(t).
\end{align*}
Combining \eqref{eq:Un_split_xy_H} and \eqref{eq:xy_split} yields the exact decomposition
\begin{align} \label{eq:U-decomp}
U_n(z,t)=S_{\epsilon,n}(z,t)+S_{\gamma,n}(z,t)+S_{\text{coef},n}(z,t)+S_{\text{cent},n}(z,t).
\end{align}

\subsubsection{Auxiliary asymptotics}

\begin{lemma} \label{lem:Brownianlimit}
Under Assumptions (A1)--(A3), on a possibly richer probability space, there exists a $C[0,1]$-valued Brownian motion $B(z,t)$ with $$\mathrm{Cov}(B(z,s),B(w,t))=\min\{z,w\}\,C(s,t),$$
where $C$ is given in \eqref{longrunCovariances}, such that, as $n \to \infty$,
$$
\sup_{z,t \in [0,1]} \bigg| \frac{1}{\sqrt n} \sum_{i=1}^{\floor{nz}} x_i(t) \epsilon_i(t) - B(z,t) \bigg| =  o_P(1).
$$
\end{lemma}

\begin{proof}
We apply Lemma~\ref{lem:SIP} to the $C[0,1]$-valued time series $Z_i:=x_i \epsilon_i$.
The exogeneity condition (A1) implies that $\E(Z_i(s)) = 0$ for all $s \in [0,1]$ and $i \in \mathbb Z$.
Note that, 
$$
[Z_i]_\eta=[x_i\epsilon_i]_\eta
\le [x_i]_\eta\|\epsilon_i\|_\infty+\|x_i\|_\infty[\epsilon_i]_\eta,
\qquad
\|Z_i\|_\infty = \|x_i \epsilon_i \|_\infty \le \|x_i\|_\infty\|\epsilon_i\|_\infty.
$$
Let $q=p/2>\max\{2,1/\eta\}$. Since $x_i=X_i-\mu_X$ with deterministic $\mu_X$ then Assumption (A2) implies that,
$\sup_{i\in\mathbb Z}\E([x_i]_\eta^p)<\infty$ and $\sup_{i\in\mathbb Z}\E(\|x_i\|_\infty^p)<\infty$.
Using $(a+b)^q\le 2^{q-1}(a^q+b^q)$, the Cauchy-Schwarz inequality, and Assumption (A2), we obtain, 
\begin{align*}
\E([Z_i]_\eta^q)
&\le 2^{q-1}\E\big([x_i]_\eta^q\|\epsilon_i\|_\infty^q\big)
   +2^{q-1}\E\big(\|x_i\|_\infty^q[\epsilon_i]_\eta^q\big) \\
&\le 2^{q-1}\E([x_i]_\eta^p)^{1/2}\E(\|\epsilon_i\|_\infty^p)^{1/2}
   +2^{q-1}\E(\|x_i\|_\infty^p)^{1/2}\E([\epsilon_i]_\eta^p)^{1/2} < \infty,
\end{align*}
and similarly $\E(\|Z_i\|_\infty^q)\le \E(\|x_i\|_\infty^p)^{1/2}\E(\|\epsilon_i\|_\infty^p)^{1/2} < \infty$.
Therefore,
$$
\sup_{i\in\mathbb Z}\E([Z_i]_\eta^q)<\infty,
\qquad
\sup_{i\in\mathbb Z}\E(\|Z_i\|_\infty^q)<\infty.
$$
Moreover, by (A3), the sequence $(Z_i)_{i \in \mathbb Z}$ is weakly stationary and $\beta$-mixing of size $-q/(q-2)$ if $\eta > 1/2$ and of size $-q(1-\eta)/(q\eta-1)$ if $\eta \leq 1/2$ because the mixing size is preserved under measurable transformations.

Then, Lemma \ref{lem:SIP} implies that on a possibly richer probability space there exists a $C[0,1]$-valued Brownian motion $\tilde B(k,t)$ with $\mathrm{Cov}(\tilde B(k,s), \tilde B(l,t)) = \min\{k,l\} C(s,t)$, where $k,l \geq 0$, $s,t \in [0,1]$, 
such that for some $\kappa > 0$,
$$
\sup_{t \in [0,1]} \left|\sum_{i=1}^{\lfloor k \rfloor} Z_i(t) - \tilde B(k,t)\right| = O(k^{1/2 - \kappa}), \qquad \text{a.s.} \quad (k \to \infty). 
$$
Define the rescaled Brownian motion as,
$$
    B(z,t) := \frac{1}{\sqrt n} \tilde B(nz,t),
$$
with the covariance structure given by,
$$
    \mathrm{Cov}(B(z_1,s), B(z_2,t)) = \frac{1}{n}\mathrm{Cov}(\tilde B(nz_1,s), \tilde B(nz_2,t)) = \min\{z_1, z_2 \} C(s,t),
$$
where it holds that,
$$
	C(s,t)=\sum_{h=-\infty}^{\infty}\mathrm{Cov}(Z_0(s),Z_h(t))
=\sum_{h=-\infty}^{\infty}\mathrm{Cov}((X_0(s)-\mu_X(s))\epsilon_0(s),(X_h(t)-\mu_X(t))\epsilon_h(t)).
$$
Then, it follows that,
\begin{align*}
	\sup_{z \in [0,1]} \sup_{t \in [0,1]} \bigg| \frac{1}{\sqrt n} \sum_{i=1}^{\floor{nz}} x_i(t) \epsilon_i(t) - B(z,t) \bigg|
	\leq \frac{1}{\sqrt n} \sup_{1 \leq k \leq n} \sup_{t \in [0,1]} \bigg|\sum_{i=1}^{k} Z_i(t) - \tilde B(k,t)  \bigg| = O(n^{-\kappa}),
\end{align*}
a.s., as $n \to \infty$, because $\sup_{k \leq n} k^{1/2-\kappa} = n^{1/2-\kappa}$, which implies,
$$
\sup_{z,t \in [0,1]} \bigg| \frac{1}{\sqrt n} \sum_{i=1}^{\floor{nz}} x_i(t) \epsilon_i(t) - B(z,t) \bigg| = O_P(n^{-\kappa}) = o_P(1).
$$

\end{proof}

\begin{lemma} \label{lem:samplemoments-asymptotics} Under Assumptions given in (A1)--(A3),  it holds as $n \to \infty$ that, 
	\begin{itemize}
		\item[(a)] $\sup_{z,t \in [0,1]} |n^{-1/2} \sum_{i=1}^{\floor{nz}} x_i(t) | = O_P(1)$
		\item[(b)] $\sup_{z,t \in [0,1]} |n^{-1/2} \sum_{i=1}^{\floor{nz}} \gamma(i/n,t) x_i(t) | = O_P(1)$
		\item[(c)] $\sup_{t \in [0,1]} |\delta_X(t)| = O_P(n^{-1/2})$
		\item[(d)] $\sup_{z,t \in [0,1]} |n^{-1/2} \sum_{i=1}^{\floor{nz}} \epsilon_i(t) | = O_P(1)$

		\item[(e)] $\sup_{z,t \in [0,1]} |n^{-1/2} \sum_{i=1}^{\floor{nz}} ((x_i(t))^2 - \sigma_X^2(t)) | = O_P(1).$
	\end{itemize}
\end{lemma}

\begin{proof} \ \\
\textit{Item (a).} 
Similarly to the proof of Lemma \ref{lem:Brownianlimit}, apply Lemma~\ref{lem:SIP} to the $C[0,1]$-valued time series $Z_i:=x_i$.
Assumptions (A2) and (A3) imply the required moment bounds and $\beta$-mixing properties for $(x_i)_{i\in\mathbb Z}$.
Hence, on a possibly richer probability space there exist $\kappa>0$ and a $C[0,1]$-valued Brownian motion $B_X(k,t)$ such that,
$$
\sup_{t\in[0,1]}\Big|\sum_{i=1}^{\lfloor k\rfloor} x_i(t)-B_X(k,t)\Big| = O(k^{1/2-\kappa}),\qquad \text{a.s.}\ (k\to\infty).
$$
Define the rescaled Brownian motion $\tilde B_X(z,t):=n^{-1/2}B_X(nz,t)$ for $z,t\in[0,1]$ with
$$
\sup_{z,t\in[0,1]}\Big|\frac{1}{\sqrt n}\sum_{i=1}^{\lfloor nz\rfloor}x_i(t)-\tilde B_X(z,t)\Big|
\le \frac{1}{\sqrt n}\sup_{1\le k\le n}\sup_{t\in[0,1]}
\Big|\sum_{i=1}^{k}x_i(t)-B_X(k,t)\Big| = o_P(1).
$$
Moreover, $\tilde B_X$ has continuous sample paths on the compact set $[0,1]^2$ and its law does not depend on $n$,
hence $\sup_{z,t \in [0,1]}|\tilde B_X(z,t)|=O_P(1)$. This implies (a). 
\newline
\textit{Item (b).}
Fix $t\in[0,1]$ and set $a_i(t):=\gamma(i/n,t)$ and $S_m(t):=\sum_{j=1}^m x_j(t)$.
By Abel's formula of summation by parts (see, e.g., Lemma 6 of \citealt{otto2022backward}), for $m\ge1$,
$$
\sum_{i=1}^{m} a_i(t)\,x_i(t)
= a_m(t)\,S_m(t) + \sum_{i=1}^{m-1}\big(a_i(t)-a_{i+1}(t)\big)\,S_i(t).
$$
Hence, with $m=\lfloor nz\rfloor$ and taking suprema,
\begin{align*}
&\sup_{z,t \in [0,1]}\Big|\frac{1}{\sqrt n}\sum_{i=1}^{\lfloor nz\rfloor}\gamma(i/n,t)x_i(t)\Big| \\
& \le \Big(\sup_{z,t \in [0,1]}|\gamma(z,t)|\Big)\,\sup_{z,t \in [0,1]}\Big|\frac{1}{\sqrt n}\sum_{i=1}^{\lfloor nz\rfloor}x_i(t)\Big| \\
&\qquad + \Big(\sup_{t \in [0,1]}\sum_{i=1}^{n-1}\big|\gamma(i/n,t)-\gamma((i+1)/n,t)\big|\Big)\,
\sup_{z,t \in [0,1]}\Big|\frac{1}{\sqrt n}\sum_{i=1}^{\lfloor nz\rfloor}x_i(t)\Big|.
\end{align*}
Because $\gamma$ is piecewise Lipschitz in $z$ uniformly over $t$ with finitely many breakpoints,
$$
\sup_{t\in[0,1]}\sum_{i=1}^{n-1}\big|\gamma(i/n,t)-\gamma((i+1)/n,t)\big|=O(1).
$$
Combining this with (a) yields (b). 
\newline
\textit{Item (c).}
Since $\delta_X(t)=n^{-1}\sum_{i=1}^n x_i(t)$, (a) gives
$$
\sup_{t\in[0,1]}|\delta_X(t)|
\le \frac{1}{\sqrt n}\sup_{z,t\in[0,1]}\Big|\frac{1}{\sqrt n}\sum_{i=1}^{\lfloor nz\rfloor}x_i(t)\Big|
=O_P(n^{-1/2}).
$$
\newline
\textit{Item (d).}
The proof is identical to (a), applied to $Z_i:=\epsilon_i$. Note that $\E(\epsilon_i(t))=0$ follows from $\E(\epsilon_i(t)| X_i)=0$.
\newline
\textit{Item (e).}
Apply Lemma~\ref{lem:SIP} to the centered $C[0,1]$-valued process, 
$$
Z_i(t):=x_i(t)^2-\E(x_0(t)^2)=x_i(t)^2-\sigma_X^2(t).
$$
Choose $q=p/2 > \max\{2,1/\eta\}$.
The required moment bounds follow from (A2) using
$\|x_i^2\|_\infty\le \|x_i\|_\infty^2$ and $[x_i^2]_\eta\le 2\|x_i\|_\infty [x_i]_\eta$,
and the $\beta$-mixing property is preserved under measurable transformations.
Then the same argument as in the proof of (a) yields that
$$
\sup_{z,t\in[0,1]}\Big|\frac{1}{\sqrt n}\sum_{i=1}^{\lfloor nz\rfloor}\big(x_i(t)^2-\sigma_X^2(t)\big)\Big|=O_P(1).
$$
\end{proof}

\subsection{Proof of Theorem \ref{thm:null_limit}}

\begin{proof}
Recall the decomposition \eqref{eq:Q-bridgestructure} and \eqref{eq:U-decomp}.
Under $\mathcal H$ we have $\gamma(i/n,t)=\gamma_0(t)$ and $\alpha(i/n,t)=\alpha_0(t)$. Hence $S_{\text{coef},n}(z,t)=0$.
Moreover, we have that,
$$
S_{\gamma,n}(z,t)=\gamma_0(t)\,\tilde V_n(z,t),
\qquad
\tilde V_n(z,t):=\frac1n\sum_{i=1}^{\lfloor nz\rfloor}x_i(t)^2.
$$
Given that $\hat x_i(t)=x_i(t)-\delta_X(t)$ then,
$$
V_n(z,t)
=\frac1n\sum_{i=1}^{\lfloor nz\rfloor}\hat x_i(t)^2
=\tilde V_n(z,t)-2\delta_X(t)\frac1n\sum_{i=1}^{\lfloor nz\rfloor}x_i(t)
+\frac{\lfloor nz\rfloor}{n}\delta_X(t)^2.
$$
We define $\Delta_n(z,t):=\tilde V_n(z,t)-V_n(z,t)$. By Lemma~\ref{lem:samplemoments-asymptotics}(a)+(c), we have that,
$$
\sup_{z,t\in[0,1]}|\Delta_n(z,t)|
\le 2\bigg( \sup_{t \in[0,1]}|\delta_X(t)| \bigg) \bigg( \sup_{z,t \in [0,1]}\Big|\frac1n\sum_{i=1}^{\lfloor nz\rfloor}x_i(t)\Big| \bigg)
+\sup_{t\in[0,1]}|\delta_X(t)|^2
=O_P(n^{-1}).
$$
Since $0\le V_n(z,t)\le V_n(1,t)$, we have $\sup_{z,t \in [0,1]}V_n(z,t)/V_n(1,t)\le 1$. Therefore it holds that,
\begin{align*}
&\sup_{z,t\in[0,1]}\Big|\sqrt n\Big(S_{\gamma,n}(z,t)-\frac{V_n(z,t)}{V_n(1,t)}S_{\gamma,n}(1,t)\Big)\Big|  \\
&= \sup_{z,t\in[0,1]}\Big|\sqrt n \gamma_0(t)\Big((V_n(z,t) + \Delta_n(z,t))-\frac{V_n(z,t)}{V_n(1,t)}(V_n(1,t) + \Delta_n(1,t))\Big)\Big| \\
&=\sup_{z,t\in[0,1]}\Big|\sqrt n\,\gamma_0(t)\Big(\Delta_n(z,t)-\frac{V_n(z,t)}{V_n(1,t)}\Delta_n(1,t)\Big)\Big|
= o_P(1),
\end{align*}
where we used boundedness of $\gamma_0$ and $\sup_{z,t \in [0,1]}|\Delta_n(z,t)|=O_P(n^{-1})$.
Therefore, using \eqref{eq:Q-bridgestructure}+\eqref{eq:U-decomp} it holds that,
\begin{align}
&\Q_n(z,t) \nonumber \\
&=\sqrt n\Big[\big(S_{\epsilon,n}(z,t)+S_{\text{cent},n}(z,t)\big)
-\frac{V_n(z,t)}{V_n(1,t)}\big(S_{\epsilon,n}(1,t)+S_{\text{cent},n}(1,t)\big)\Big] + o_P(1), \label{eq:Q-null}
\end{align}
uniformly in $(z,t)$. By Lemma \ref{lem:Brownianlimit} it holds that,
\begin{align}\label{eq:Q-null-1}
\sup_{z,t\in[0,1]}\big|\sqrt n\,S_{\epsilon,n}(z,t)-B(z,t)\big|=o_P(1),
\end{align}
and in particular $\sup_{t\in[0,1]}|\sqrt n\,S_{\epsilon,n}(1,t)|=O_P(1)$. Next, write $\hat x_i(t)^2=x_i(t)^2-2\delta_X(t)x_i(t)+\delta_X(t)^2$. Using
Lemma \ref{lem:samplemoments-asymptotics}(a) and (c) we obtain that,
\begin{align}
\sup_{z,t\in[0,1]}\Big|V_n(z,t)-\frac1n\sum_{i=1}^{\lfloor nz\rfloor}x_i(t)^2\Big|=O_P(n^{-1}).
\end{align}
Together with Lemma \ref{lem:samplemoments-asymptotics}(e) this yields,
\begin{align}
\sup_{z,t\in[0,1]}\big|V_n(z,t)-z\sigma_X^2(t)\big|=O_P(n^{-1/2}),
\qquad
\sup_{t\in[0,1]}\big|V_n(1,t)-\sigma_X^2(t)\big|=O_P(n^{-1/2}).
\end{align}
Since $\inf_{t\in[0,1]}\sigma_X^2(t)>0$ by (A1), we have $\sup_{t\in[0,1]}|V_n(1,t)^{-1}|=O_P(1)$ and hence we have 
\begin{align} 
	&\sup_{z,t \in [0,1]}\Big|\frac{V_n(z,t)}{V_n(1,t)}-z\Big| \notag \\
&\quad \le \sup_{z,t \in [0,1]}\frac{|V_n(z,t)-zV_n(1,t)|}{V_n(1,t)}  \notag \\
&\quad \le \Big(\sup_{t\in [0,1]} V_n(1,t)^{-1}\Big)\Big(\sup_{z,t\in [0,1]}|V_n(z,t)-z\sigma_X^2(t)|
+\sup_{t\in [0,1]}|V_n(1,t)-\sigma_X^2(t)|\Big)  \notag \\
&\quad = o_P(1).  \label{eq:Q-null-2}
\end{align}
Finally, under $\mathcal H$ we have $y_i(t)=\gamma_0(t)x_i(t)+\epsilon_i(t)$ and
$\delta_Y(t)=\gamma_0(t)\delta_X(t)+n^{-1}\sum_{i=1}^n\epsilon_i(t)$, so that Lemma
\ref{lem:samplemoments-asymptotics}(a)+(c)+(d) implies,
\begin{align}\label{eq:Q-null-3}
\sup_{z,t\in[0,1]}|S_{\text{cent},n}(z,t)|=O_P(n^{-1}).
\end{align}
Combining \eqref{eq:Q-null}--\eqref{eq:Q-null-3}, we then have,
\begin{align*}
\Q_n(z,t)-\big(B(z,t)-zB(1,t)\big)
&=\big(\sqrt n S_{\epsilon,n}(z,t)-B(z,t)\big)
-z\big(\sqrt n S_{\epsilon,n}(1,t)-B(1,t)\big) \\
&\quad -\Big(\frac{V_n(z,t)}{V_n(1,t)}-z\Big)\sqrt n S_{\epsilon,n}(1,t) \\
&\quad +\sqrt n\Big(S_{\text{cent},n}(z,t)-\frac{V_n(z,t)}{V_n(1,t)}S_{\text{cent},n}(1,t)\Big).
\end{align*}
Taking suprema over $(z,t)\in[0,1]^2$, the first two terms are $o_P(1)$ by \eqref{eq:Q-null-1},
the third is $o_P(1)$ by \eqref{eq:Q-null-2} and $\sup_{t \in [0,1]}|\sqrt n S_{\epsilon,n}(1,t)|=O_P(1)$,
and the last is $o_P(1)$ by equations \eqref{eq:Q-null-2} and \eqref{eq:Q-null-3}. This proves that,
$$
\sup_{z,t\in[0,1]}\Big|\Q_n(z,t)-\big(B(z,t)-zB(1,t)\big)\Big|=o_P(1).
$$
\end{proof}

\subsection{Proof of Theorem \ref{thm:alt_limit}}

\begin{proof}
Recall the decomposition \eqref{eq:Q-bridgestructure} and \eqref{eq:U-decomp}.
We start with proving the asymptotic behavior of the three partial sum processes $S_{\epsilon,n}$, $S_{\text{coef},n}$ and $S_{\text{cen },n}$ from \eqref{eq:U-decomp}.
For $S_{\epsilon,n}$, Lemma~\ref{lem:Brownianlimit} yields
$$
\sup_{z,t\in[0,1]}\big|\sqrt n\,S_{\epsilon,n}(z,t)\big|=O_P(1).
$$
For $S_{\text{coef},n}$,
since $b(z,t):=\alpha(z,t)-\bar\alpha_n(t)+\mu_X(t)\big(\gamma(z,t)-\bar\gamma_n(t)\big)$
is continuous in $t$ and piecewise Lipschitz in $z$ uniformly over $t$, $\sup_{t\in[0,1]}\sum_{i=1}^{n-1}|b((i+1)/n,t)-b(i/n,t)|=O(1)$, and $\sup_{z,t \in [0,1]}|b(z,t)|<\infty$.
Applying Abel's summation by parts as in the proof of Lemma~\ref{lem:samplemoments-asymptotics}(b), with $a_i(t)=b_i(t)$ and taking partial sums of $x_i(t)$, and using Lemma~\ref{lem:samplemoments-asymptotics}(a) together
imply that $\sup_{z,t\in[0,1]}|\frac{1}{\sqrt n}\sum_{i=1}^{\lfloor nz\rfloor}b(i/n,t)\,x_i(t)|=O_P(1)$. Hence,
$$
\sup_{z,t \in [0,1]}\big|\sqrt n\,S_{\text{coef},n}(z,t)\big|=O_P(1).
$$
For $S_{\text{cent},n}$, note that $Y_i(t)-\E(Y_i(t))=\gamma(i/n,t)\,x_i(t)+\epsilon_i(t)$, so
$$
\delta_Y(t)=\hat\mu_Y(t)-\bar\mu_Y(t)=\frac1n\sum_{i=1}^n\big(Y_i(t)-\E(Y_i(t))\big)
=\frac1n\sum_{i=1}^n\gamma(i/n,t)\,x_i(t)+\frac1n\sum_{i=1}^n\epsilon_i(t).
$$
Therefore, by Lemma~\ref{lem:samplemoments-asymptotics}(b)+(d), $\sup_{t\in[0,1]}|\delta_Y(t)|=O_P(n^{-1/2})$.
Using this together with Lemma~\ref{lem:samplemoments-asymptotics}(a)+(c)+(d) and 
$y_i(t)=\alpha_i^*(t)+\mu_X(t)\gamma_i^*(t)+\gamma(i/n,t)x_i(t)+\epsilon_i(t)$ yields
$$
\sup_{z,t\in[0,1]}\Big|\frac1n\sum_{i=1}^{\lfloor nz\rfloor}y_i(t)\Big|=O_P(1),
\qquad
\sup_{z,t\in[0,1]}\Big|\frac1n\sum_{i=1}^{\lfloor nz\rfloor}x_i(t)\Big|=O_P(n^{-1/2}),
$$
and hence from \eqref{eq:S-cent},
$$
\sup_{z,t\in[0,1]}|\sqrt n\,S_{\text{cent},n}(z,t)|=O_P(1).
$$
Combining the bounds gives
$$
	\sup_{z,t\in[0,1]}\Big|\sqrt n\big(S_{\epsilon,n}(z,t)+S_{\text{coef},n}(z,t)+S_{\text{cent},n}(z,t)\big)\Big| = O_P(1)
$$
and, by \eqref{eq:U-decomp},
\begin{align} \label{eq:U-approx-alt}
	\sup_{z,t\in[0,1]} \sqrt n \, \Big| U_n(z,t) - S_{\gamma,n}(z,t) \Big| = O_P(1).
\end{align}
For $S_{\gamma,n}$, write
$$
S_{\gamma,n}(z,t)=\frac1n\sum_{i=1}^{\lfloor nz\rfloor}\gamma(i/n,t)\,x_i(t)^2
=\sigma_X^2(t)\,\frac1n\sum_{i=1}^{\lfloor nz\rfloor}\gamma(i/n,t)
+\frac1n\sum_{i=1}^{\lfloor nz\rfloor}\gamma(i/n,t)\big(x_i(t)^2-\sigma_X^2(t)\big).
$$
Since $\gamma(\cdot,t)$ is piecewise Lipschitz uniformly over $t$, the first term is a Riemann sum and
$$
\sup_{z,t\in[0,1]}\Big|\frac1n\sum_{i=1}^{\lfloor nz\rfloor}\gamma(i/n,t)-\int_0^z\gamma(u,t)\,du\Big|=O(n^{-1}).
$$
For the second term, apply Abel summation as in Lemma~\ref{lem:samplemoments-asymptotics}(b) with
$a_i(t)=\gamma(i/n,t)$ and taking partial sums of $(x_i(t)^2-\sigma_X^2(t))$, and useing Lemma~\ref{lem:samplemoments-asymptotics}(e) together imply
$$
\sup_{z,t\in[0,1]}\Big|\frac{1}{\sqrt n}\sum_{i=1}^{\lfloor nz\rfloor}\gamma(i/n,t)\big(x_i(t)^2-\sigma_X^2(t)\big)\Big|=O_P(1).
$$
Together the above steps yield that
$$
\sup_{z,t\in[0,1]}\Big|\sqrt n\Big(S_{\gamma,n}(z,t)-\sigma_X^2(t)\int_0^z\gamma(u,t)\,du\Big)\Big|=O_P(1).
$$
and with \eqref{eq:U-approx-alt},
\begin{align}\label{eq:U-approx-alt1}
\sup_{z,t\in[0,1]}\Big|\sqrt n\Big(U_n(z,t)-\sigma_X^2(t)\int_0^z\gamma(u,t)\,du\Big)\Big|=O_P(1).
\end{align}
Recall from \eqref{eq:Q-bridgestructure} that
$$
\Q_n(z,t)=\sqrt n\Big(U_n(z,t)-\frac{V_n(z,t)}{V_n(1,t)}U_n(1,t)\Big).
$$
Moreover, by the same argument as in \eqref{eq:Q-null-2},
$$
\sup_{z,t \in [0,1]} \sqrt n\, \Big|\frac{V_n(z,t)}{V_n(1,t)}-z\Big|=O_P(1)
$$
also holds under $\mathcal K$.
Therefore, with \eqref{eq:U-approx-alt1},
\begin{align*}
\Q_n(z,t)
&=\sqrt n\,\sigma_X^2(t)\Big(\int_0^z\gamma(u,t)\,du-\frac{V_n(z,t)}{V_n(1,t)}\int_0^1\gamma(u,t)\,du\Big)+O_P(1) \\
&=\sqrt n\,\sigma_X^2(t)\Big(\int_0^z\gamma(u,t)\,du-z\int_0^1\gamma(u,t)\,du\Big)+O_P(1),
\end{align*}
uniformly over $(z,t)\in[0,1]^2$.
Finally, the assertion follows since,
$$
\int_0^z\gamma(u,t)\,du-z\int_0^1\gamma(u,t)\,du=\int_0^z\gamma^*(u,t)\,du,
\qquad
\gamma^*(z,t)=\gamma(z,t)-\int_0^1\gamma(u,t)\,du.
$$
\end{proof}

\newpage

\bibliography{flrBib}

\end{document}